\documentclass{article}

\usepackage{graphicx} 
\usepackage{fullpage}
\usepackage[top=2cm, bottom=4.5cm, left=2.5cm, right=2.5cm]{geometry}
\usepackage{amsmath,amsthm,amsfonts,amssymb,amscd}
\usepackage{lastpage}
\usepackage{enumerate}
\usepackage{mathrsfs}
\usepackage{xcolor}
\usepackage{hyperref}
\usepackage{tikz}

\usepackage[noadjust]{cite}
\usepackage[noend, boxed]{algorithm2e}
\usepackage{placeins}
\usepackage[nobreak]{mdframed}

\usetikzlibrary{positioning, shapes.geometric, shapes.misc}
\tikzset{main node/.style = {circle, fill = white, font = \fontsize{8}{0}\selectfont, draw, minimum size=1.3cm, inner sep=0pt}}
\newtheorem{theorem}{Theorem}

\newtheorem{corollary}[theorem]{Corollary}
\newtheorem{lemma}[theorem]{Lemma}

\newtheorem{definition}{Definition}

\newcommand{\len}{\texttt{len}}
\usepackage{todonotes}

\newcommand{\Oish}{\widetilde{O}}
\newcommand{\Ohat}{\widehat{O}}
\newcommand{\Omegaish}{\widetilde{\Omega}}
\newcommand{\Omegahat}{\widehat{\Omega}}

\newcommand{\dist}{\texttt{dist}}
\newcommand{\vdist}{\texttt{vdist}}
\newcommand{\epoch}{\texttt{epoch}}

\newcommand{\val}{\texttt{val}}

\newcommand{\eps}{\varepsilon}

\newcommand{\poly}{\texttt{poly}}
\newcommand{\vfg}{\texttt{vfg}}
\newcommand{\efg}{\texttt{efg}}
\newcommand{\vufg}{\texttt{vufg}}
\newcommand{\vuufg}{\texttt{vuufg}}

\newcommand{\mass}{\texttt{mass}}

\usepackage{makecell}
\usepackage{multirow}

\usetikzlibrary{arrows.meta}

\newcommand{\mvc}{\textsc{VertexCut}}

\newcommand{\cost}{\texttt{cost}}

\author{Greg Bodwin and Luba Samborska\\University of Michigan EECS\\ \texttt{\{bodwin, liubovs\}@umich.edu}}
\date{}

\begin{document}
\title{Improved Upper Bounds for the Directed Flow-Cut Gap\thanks{This work was supported by NSF:AF 2153680}}
\maketitle





\thispagestyle{empty}
\begin{abstract}
We prove that the flow-cut gap for $n$-node directed graphs is at most $n^{1/3 + o(1)}$.
This is the first improvement since a previous upper bound of $\Oish(n^{11/23})$ by Agarwal, Alon, and Charikar (STOC '07), and it narrows the gap to the current lower bound of $\Omegaish(n^{1/7})$ by Chuzhoy and Khanna (JACM '09).
We also show an upper bound on the directed flow-cut gap of $W^{1/2}n^{o(1)}$, where $W$ is the sum of the minimum fractional cut weights.

As auxiliary contributions along the way to our main proof, we prove near-equivalence between the edge and vertex directed flow-cut gaps, and we show that one can assume unit capacities and uniform fractional cut weights without loss of generality when parametrizing by $W$.
\end{abstract}

\clearpage
\setcounter{page}{1}

\section{Introduction}

We study multi-pair extensions of the \emph{Min-Cut Max-Flow (MCMF) Theorem}, a primitive in graph theory and algorithms.
This theorem states that, for any graph $G$ and node pair $(s, t)$, the maximum value of an $s \leadsto t$ flow is equal to the minimum capacity of a cut separating $s$ and $t$.
In our setting, instead of one input pair $(s, t)$, the goal is to analyze flows or cuts among an entire set of demand pairs $P \subseteq V(G) \times V(G)$.
Even in this generalized setting, it remains true that the capacity of the minimum \emph{fractional} multicut that simultaneously separates all pairs in $P$ is equal to the value of the maximum multiflow\footnote{In the maximum multiflow problem, the goal is to simultaneously push any amount of flow between each demand pair (while respecting capacities), and the goal is to maximize the sum of flow values over the demand pairs.} for $P$ (this follows by LP duality; see e.g., \cite{CK09} Section 2 for formal details).
However, it is no longer generally true that the minimum \emph{integral} multicut for $P$ has the same capacity.
An example of an instance on which these quantities differ is shown in Figure \ref{fig:mcmffail}.


\begin{figure}
\begin{center}
\begin{tikzpicture}[every node/.style={font=\small}]

  \node at (0, -3) {Integral Multicut};

  \node[draw, circle, fill=red!15, inner sep=1.5pt] (s1) at (-1.8, -2.2) {$s_1$};
  \node[draw, rectangle, rounded corners=1pt, fill=blue!15, inner sep=2pt] (s2) at ( 0.0, -2.2) {$s_2$};
  \node[draw, diamond, aspect=1.2, fill=green!15, inner sep=1.2pt] (s3) at ( 1.8, -2.2) {$s_3$};

  \node[draw, diamond, aspect=1.2, fill=green!15, inner sep=1.2pt] (t3) at (-1.8,  2.2) {$t_3$};
  \node[draw, rectangle, rounded corners=1pt, fill=blue!15, inner sep=2pt] (t2) at ( 0.0,  2.2) {$t_2$};
  \node[draw, circle, fill=red!15, inner sep=1.5pt] (t1) at ( 1.8,  2.2) {$t_1$};

  \node[fill=black, circle, inner sep=1.6pt] (a) at (-0.8,  0.6) {};
  \node[fill=black, circle, inner sep=1.6pt]  (b) at ( 0.8,  0.6) {};
  \node[fill=black, circle, inner sep=1.6pt] (c) at ( 0.0, -0.6) {};

  \draw[->] (a) to node [midway, red] {\bf \Huge $\times$} (b);
  \draw[->] (b) to node [midway, red] {\bf \Huge $\times$} (c);
  \draw[->] (c) -- (a);
  
  \draw[->] (s1) -- (a);
  \draw[->] (c) to [bend right=50] (t1);

  \draw[->] (s2) -- (c);
  \draw[->] (b) -- (t2);

  \draw[->] (s3) -- (b);
  \draw[->] (a) -- (t3);
\end{tikzpicture} \hspace{0.08\textwidth}%
\begin{tikzpicture}[every node/.style={font=\small}]

  \node at (0, -3) {Fractional Multicut};

  \node[draw, circle, fill=red!15, inner sep=1.5pt] (s1) at (-1.8, -2.2) {$s_1$};
  \node[draw, rectangle, rounded corners=1pt, fill=blue!15, inner sep=2pt] (s2) at ( 0.0, -2.2) {$s_2$};
  \node[draw, diamond, aspect=1.2, fill=green!15, inner sep=1.2pt] (s3) at ( 1.8, -2.2) {$s_3$};

  \node[draw, diamond, aspect=1.2, fill=green!15, inner sep=1.2pt] (t3) at (-1.8,  2.2) {$t_3$};
  \node[draw, rectangle, rounded corners=1pt, fill=blue!15, inner sep=2pt] (t2) at ( 0.0,  2.2) {$t_2$};
  \node[draw, circle, fill=red!15, inner sep=1.5pt] (t1) at ( 1.8,  2.2) {$t_1$};

  \node[fill=black, circle, inner sep=1.6pt] (a) at (-0.8,  0.6) {};
  \node[fill=black, circle, inner sep=1.6pt]  (b) at ( 0.8,  0.6) {};
  \node[fill=black, circle, inner sep=1.6pt] (c) at ( 0.0, -0.6) {};

  \draw[->] (a) to node [midway, above=-0.1cm, red, sloped] {\footnotesize $1/2$} (b);
  \draw[->] (b) to node [midway, below=-0.1cm, red, sloped] {\footnotesize $1/2$} (c);
  \draw[->] (c) to node [midway, below=-0.1cm, red, sloped] {\footnotesize $1/2$} (a);
  
  \draw[->] (s1) -- (a);
  \draw[->] (c) to [bend right=50] (t1);

  \draw[->] (s2) -- (c);
  \draw[->] (b) -- (t2);

  \draw[->] (s3) -- (b);
  \draw[->] (a) -- (t3);
\end{tikzpicture} \hspace{0.08\textwidth}%
\begin{tikzpicture}[every node/.style={font=\small}]
  \node[draw, circle, fill=red!15, inner sep=1.5pt] (s1) at (-1.8, -2.2) {$s_1$};
  \node[draw, rectangle, rounded corners=1pt, fill=blue!15, inner sep=2pt] (s2) at ( 0.0, -2.2) {$s_2$};
  \node[draw, diamond, aspect=1.2, fill=green!15, inner sep=1.2pt] (s3) at ( 1.8, -2.2) {$s_3$};

  \node[draw, diamond, aspect=1.2, fill=green!15, inner sep=1.2pt] (t3) at (-1.8,  2.2) {$t_3$};
  \node[draw, rectangle, rounded corners=1pt, fill=blue!15, inner sep=2pt] (t2) at ( 0.0,  2.2) {$t_2$};
  \node[draw, circle, fill=red!15, inner sep=1.5pt] (t1) at ( 1.8,  2.2) {$t_1$};

  \node[fill=black, circle, inner sep=1.6pt] (a) at (-0.8,  0.6) {};
  \node[fill=black, circle, inner sep=1.6pt]  (b) at ( 0.8,  0.6) {};
  \node[fill=black, circle, inner sep=1.6pt] (c) at ( 0.0, -0.6) {};

  
  \draw[->, red] (s1) to node [left, red] {$\frac{1}{2}$} (a);
  \draw[->, red] (c) to [bend right=50] (t1);

  \draw[->, blue] (s2) to node [left, blue] {$\frac{1}{2}$} (c);
  \draw[->, blue] (b) -- (t2);

  \draw[->, green!50!black] (s3) to node [left, green!50!black] {$\frac{1}{2}$} (b);
  \draw[->, green!50!black] (a) -- (t3);

  \draw [red] (a) to [bend left=5] (b) to [bend right=5] (c);
  \draw [blue] (c) to [bend left=5] (a) to [bend right=5] (b);
  \draw [green!50!black] (b) to [bend left=5] (c) to [bend right=5] (a);

  \node at (0, -3) {Multiflow};

\end{tikzpicture}
\end{center}
\vspace{-0.2cm} 
\caption{\label{fig:mcmffail} An instance $(G, P)$ with unit edge costs/capacities, on which the minimum fractional and integral cuts differ.  A minimum integral cut has cost $2$ (left), while a minimum fractional cut has cost $\frac{3}{2}$ (middle), matching the maximum flow of value $\frac{3}{2}$ (right).}
\end{figure}
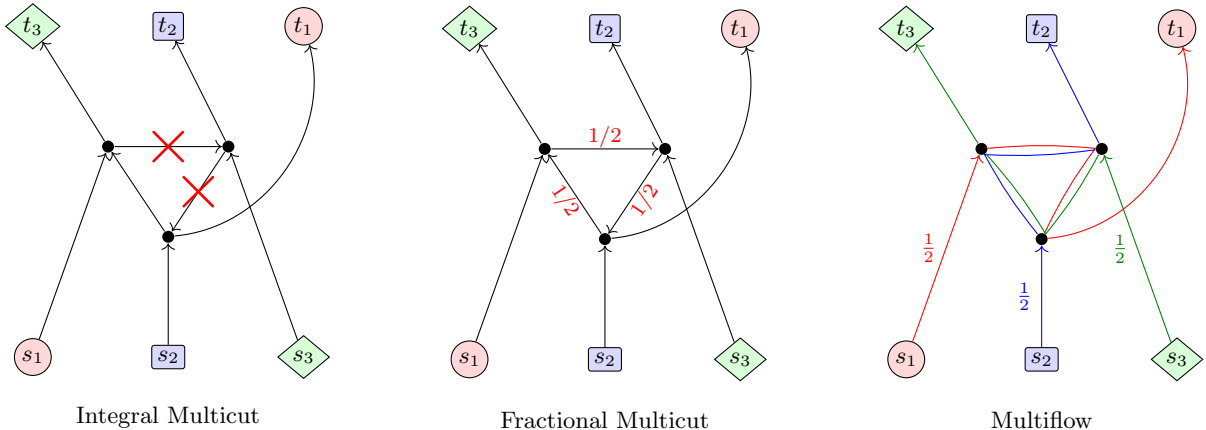

\renewcommand{\arraystretch}{2}
\begin{table}[t]
\begin{center}
\begin{tabular}{ccc}
\makecell{\textbf{Flow-Cut}\\\textbf{Gap}} & \textbf{Notes} & \textbf{Citations}\\
\hline
$O(\log |P|)$ & undirected only & \cite{LR99, GVY96}\\
\hline
\hline
$O(|P|)$ & MCMF Theorem & \cite{FF15}\\
$\Oish(n^{1/2})$ &  & \cite{CKR01}\\
$\min\left\{O(W), O(n^{1/2})\right\}$ & & \cite{Gupta03}\\
$\Oish(n^{11/23})$ & & \cite{AAC07}\\
$\min\left\{\Ohat(n^{1/3}), \Ohat(W^{1/2})\right\}$ &  & \color{red} \textbf{this paper}\\
\hline
\hline
$\Omega(\log |P|)$ &  & \cite{LR99, GVY96}\\
$\Omega(|P|)$ & $|P| \le O\left( \frac{\log n}{\log \log n} \right)$ only & \cite{SSZ04}\\
$\Omegaish(n^{1/7})$ &  \text{directed only} & \cite{CK09} (see also \cite{BHT23})
\end{tabular}
\end{center}
\caption{\label{tbl:fcgprior} Previous work on the flow-cut gap in the setting of $n$-node input graphs, with demand pairs $P$ and sum of fractional cut weights $W$.}
\end{table}

A fundamental question in the area is to determine the maximum possible ratio between these quantities in the multi-pair setting.
This ratio is called the \emph{flow-cut gap}.
The flow-cut gap has been abstracted and studied since at least the late '80s; we summarize some of the progression of bounds in Table \ref{tbl:fcgprior}.
A highlight is the seminal work of Leighton and Rao \cite{LR99}, which essentially closed the problem for \emph{undirected} $n$-node input graphs: there are upper and lower bounds of the form $\Theta(\log n)$.
Unfortunately, this understanding does not extend to the setting of \emph{directed} input graphs.
In the current state-of-the-art, Chuzhoy and Khanna \cite{CK07, CK09} constructed directed $n$-node graphs where the flow-cut gap is $\Omegaish(n^{1/7})$, providing a strong separation from the logarithmic gap known for undirected graphs.
Meanwhile, the best upper bound is $\Oish(n^{11/23})$ by Agarwal, Alon, and Charikar \cite{AAC07}, leaving a significant gap in the exponent.
Despite extensive research attention from the community in the two decades since these results (e.g., \cite{Filtser24, FFIKLMZ25, CSW13, LR13, OS81, KLR19, LMM13, Rao99, GNRS04, SSS17, KS22}; see also the 2014 survey of Chuzhoy \cite{Chuzhoy14}), these bounds have resisted improvement in the general case; better upper and lower bounds are known only for various specific graph classes such as planar, minor-free, bounded-treewidth, etc.

Our main result is the following new upper bound, lowering the exponent in the upper bound to $1/3$.
Here and throughout the paper, we use $\Ohat$ notation to hide $n^{o(1)}$ or $n^{-o(1)}$ factors, where $n$ is the number of nodes in the relevant graph (and $\Omegahat$ notation is used similarly).

\begin{mdframed}[backgroundcolor=green!10]
\begin{theorem} [Main Result] \label{thm:main}
The flow-cut gap for $n$-node directed graphs is at most $\Ohat\left(n^{1/3}\right)$.
This holds in both the edge- and vertex-capacitated settings.
\end{theorem}
\end{mdframed}

For a more formal statement of this theorem, see Section \ref{sec:formalisms}.
Our proof is constructive, in that we provide an algorithm that runs in randomized polynomial time, and (deterministically) produces a valid cut that (with high probability) satisfies the size/cost bound implied by this theorem.

\subsection{Parametrizing the Flow-Cut Gap by Total Weight}

The technical part of this paper will focus on a different parametrization of the flow-cut gap, by the \emph{total fractional cut weight $W$} of the instance.
That is: given a directed graph $G$ and set of demand pairs $P$, letting $w$ be a minimum fractional multicut for $P$, the total fractional cut weight $W$ is defined to be the sum of the cut weights used by $w$.
We show the following bound:
\begin{mdframed}[backgroundcolor=green!10]
\begin{theorem} \label{thm:mainweight}
The flow-cut gap for directed graphs with total fractional cut weight $W$ is $\Ohat(W^{1/2})$.
This holds in both the edge- and vertex-capacitated settings.
\end{theorem}
\end{mdframed}

(See Section \ref{sec:formalisms} for a more precise version of this theorem statement.)
For comparison to prior work, the argument of Gupta \cite{Gupta03} implies a bound of $O(W)$.
It might also be possible to extract similar bounds parametrized by $W$ from the arguments in other prior work \cite{CKR01, AAC07}, but this seems less clear.

We have three reasons to highlight this parametrization by $W$.
First, on philosophical grounds, we should expect flow-cut gap bounds to be invariant to the addition of isolated nodes, or more generally any nodes that do not affect connectivity among the demand pairs.
These nodes would affect the parameter $n$ but not $W$, suggesting that the parameter $W$ might more accurately reflect the hardness of an instance.
Second, in the vertex setting, bounds for $W$ \emph{imply} bounds for $n$ via the following well-known reduction:
\begin{mdframed}[backgroundcolor=green!10]
\begin{theorem} [Folklore; see Theorem \ref{thm:vfgself}] \label{thm:introwton}
Suppose that the vertex flow-cut gap for $n$-node, $W$-weight directed graphs is at most $\Ohat(W^c)$, for some constant $c$.
Then it is also at most $\Ohat(n^{\frac{c}{1+c}})$.
\end{theorem}
\end{mdframed}

This theorem statement is essentially an abstraction of a standard technique in the area, where one cuts all nodes of sufficiently high weight as a preprocessing step.
Through this reduction, the vertex version of Theorem \ref{thm:mainweight} \emph{implies} the vertex version of Theorem \ref{thm:main}.
We note that a converse reduction, using a flow-cut gap bound depending on $n$ to imply one depending on $W$, does not seem to be known or implied by current techniques.

Third, as a new technical result, we develop a useful simplification of the problem which seems to be available only when parametrizing by $W$.
It is easiest to understand the flow-cut gap in the special case of \emph{uncapacitated} input instances (where all nodes/edges have capacity $1$), and also on instances of \emph{uniform weight} (where all nodes/edges have the same fractional cut weight as each other).
When parametrizing by $n$ these simplifications seem to lose generality, and so prior work has had to contend with variable capacities or weights.\footnote{In the edge multigraph setting, it is trivial to remove capacities, by scaling such that capacities are integers and then replacing each edge with a number of parallel edges equal to its capacity.  However, (1) this reduction can blow up the total cut weight $W$, and (2) the analogous reduction for the vertex setting can blow up the number of nodes $n$.  It is also not known how to reduce to uniform weights in either setting, when parametrizing by $n$.}
However, we prove that these simplifications can be made without loss of generality when parametrizing on $W$, in the vertex setting:

\begin{mdframed}[backgroundcolor=green!10]
\begin{theorem} [See Section \ref{sec:reductions}] \label{thm:utwreductions}
Suppose that the directed vertex flow-cut gap is at most $\Ohat(W^c)$, for some constant $c$, in the special case of uncapacitated input graphs in which all nodes have the same weight as each other.
Then a bound of $\Ohat(W^c)$ also holds for the directed vertex flow-cut gap in the general case (with arbitrary capacities and weights).
\end{theorem}
\end{mdframed}

The previous theorem only applies in the vertex setting; it remains conceivable that the directed edge flow-cut gap is truly different in the general vs.\ uncapacitated/uniform weight settings, even when parametrized on $W$.
However, we show an additional reduction that lifts bounds from the vertex setting to the edge setting, and so it is enough for our purposes to focus on the vertex setting where these reductions apply.
The vertex-to-edge direction of the following reduction is folklore, while the edge-to-vertex direction is new.

\begin{mdframed} [backgroundcolor=green!10]
\begin{theorem} [See Theorems \ref{thm:edgetovert} and \ref{thm:vfgleefg}] \label{thm:edgetovertintro}
The directed edge and vertex flow-cut gaps are equivalent, up to factors of $O(\log n)$ in the number of nodes $n$ and $O(1)$ in the total weight $W$.
\end{theorem}
\end{mdframed}


\subsection{Relationship to Nearby Problems}

We next overview some nearby problems that are known to be related to the flow-cut gap.

\paragraph{Multicut Approximation Algorithms.}
Given a multicut instance $G = (V, E, \cost), P$, it is NP-hard to find an (edge or vertex) integral cut $X$ of minimum $\cost(X)$.
In fact, various hardness of approximation results are known \cite{CKKRS06, CK09}; the highest lower bound is $\Omega( 2^{\log^{1-\eps} n})$, which assumes $\textsc{np} \not \subseteq \textsc{zpp}$.


The directed flow-cut gap provides a natural strategy for approximation algorithms, by computing a fractional cut $w$ with an LP and then rounding to $X$.
Most of the prior work on upper bounds for the problem work this way (a notable exception is \cite{KKN05}).
Indeed, since our results are algorithmic, they imply:
\begin{mdframed}[backgroundcolor=green!10]
\begin{corollary}
There is a randomized polynomial-time algorithm for minimum-cost edge or vertex multicut that achieves an approximation ratio of $\min\left\{\Ohat(n^{1/3}), \Ohat(W^{1/2})\right\}$.
\end{corollary}
\end{mdframed}

Like the flow-cut gap, the previous upper bound for this approximation ratio was $\Oish(n^{11/23})$ by Agarwal, Alon, and Charikar \cite{AAC07}.

\paragraph{The Sparsest Cut Gap.}
The \emph{sparsest multicut} computational problem is a relaxation of minimum multicut.
Here, the cut $X$ does not necessarily have to separate \emph{all} demand pairs.
The goal is to minimize the ratio of $\cost(X)$ to the number of demand pairs separated.
Its fractional version is the LP dual of \emph{maximum concurrent flow}, where the goal is to maximize the \emph{minimum} flow value pushed between any demand pair (rather than the sum of flow values).
One can again ask about the \emph{sparsest cut gap}, which is the maximum possible ratio between these quantities when the cut is forced to be integral.

Following some initial work on this problem, e.g.~\cite{HR06}, it was proved by Agarwal, Alon, and Charikar \cite{AAC07} that in the setting of general directed graphs, this gap is the same as the flow-cut gap up to $O(\log n)$ factors.
Their reduction preserves edge weights, and hence we have:
\begin{mdframed}[backgroundcolor=green!10]
\begin{corollary}
The sparsest cut gap is at most $\min\left\{\Ohat(n^{1/3}), \Ohat(W^{1/2})\right\}$.
\end{corollary}
\end{mdframed}

\paragraph{(Weak) Low-Diameter Decompositions.}

A \emph{(weak) low-diameter decomposition} (LDD), also sometimes called a \emph{random quasipartition} \cite{MSS18, AF26}, is a randomized algorithm that takes on input an edge-weighted graph $G = (V, E, w)$ and a parameter $\Delta$, and returns a cut $X \subseteq E$ with the following properties:
\begin{itemize}
\item Every pair of nodes with $\dist_w(s, t) \ge \Delta$ is cut by $X$ (deterministically),\footnote{A \emph{strong} LDD instead requires $\dist_{G \setminus X}(s, t) \ge \Delta$.} and
\item For any edge $e \in E$, the probability that $e \in X$ is at most $\frac{w(e)}{\Delta} \cdot \texttt{LOSS}$, for some hopefully-small expression \texttt{LOSS}.\footnote{Some LDDs allow an additional $+\frac{1}{\poly(n)}$ term in the edge-cut probability, where algorithmic applications can tolerate it.  We will not consider such a term in our discussion.}
\end{itemize}

LDDs are a basic tool in graph algorithms with applications to metric embeddings, approximation algorithms, and shortest path computation (e.g., \cite{Bartal96, FRT03, HHWZ25, BFHL25, BCF23, BNW25}).
They are also a natural \emph{strengthening} of the flow-cut gap.
To see this: by rescaling edge weights, we may assume without loss of generality that $\Delta=1$, and so the first property states that the weights $w$ form a fractional cut for the node pairs that we wish to separate.
Applying linearity of expectation, the second property implies that in expectation, the size (or cost) of the cut $X$ is at most $\texttt{LOSS}$ times that of the fractional cut $w$.
Thus, it implies \texttt{LOSS} as a bound on the flow-cut gap.

Recently, Bringmann, Fischer, Haeupler, and Latypov \cite{BFHL25} (Lemma 2.2) showed that the multiplicative weights update method can be used to show the \emph{converse} reduction: bounds on the flow-cut gap imply bounds on weak LDDs.
Although the focus of \cite{BFHL25} is on edge cuts and roundtrip distance, their proof generalizes directly to our setting, implying: 

\begin{mdframed}[backgroundcolor=green!10]
\begin{corollary} \label{cor:wldd}
Every $n$-node directed graph $G = (V, E, w)$ has a weak LDD with parameter $\texttt{LOSS} \le \min\left\{\Ohat(n^{1/3}), \Ohat(W^{1/2})\right\}$.
\end{corollary}
\end{mdframed}

Since there are some mild technical details needed to update their proof to our setting, we provide a formal proof of this corollary in Section \ref{sec:wlddred}.
In fact, we show a more general reduction: any bound on flow-cut gaps implies the corresponding bound for weak LDDs (up to log factors).

To clarify a possible point of confusion: there has been a recently popular line of work on ``directed LDDs,'' with applications to shortest path algorithms under negative weights \cite{HHWZ25, BFHL25, BCF23, BNW25}.
These papers are specifically studying LDDs for the \emph{directed roundtrip} setting.
Through the previous reduction, this is equivalent to the flow-cut gap in a model where for each demand pair $(s, t) \in P$, (1) the fractional cut $w$ must satisfy $\dist_w(s, t) + \dist_w(t, s) \ge 1$, and (2) the integral cut $X$ must hit all $s \leadsto t \leadsto s$ ``roundtrip'' paths.
This line of work has nearly closed the flow-cut gap for roundtrip distances: the state-of-the-art upper bound is $O(\log n \log \log n)$, and the lower bound is $\Omega(\log n / \log \log n)$.
The previous corollary gives a directed LDD in the one-way (non-roundtrip) setting, where $\texttt{LOSS}$ will need to be polynomial in $n$ in the worst case \cite{CK09}.

\subsection{Organization}

In Section \ref{sec:formalisms}, we give more formal statements of our main results, and apply some reductions (which we prove in Section \ref{sec:reductions}) that allow us to reduce the problem to a simplified main lemma (Lemma \ref{lem:vcr}).
After these setup steps, we provide a technical overview of the proof of our main lemma in Section \ref{sec:intuition}, intended to provide structure and intuition for some of the more involved parts of our argument.

We give the main algorithm in Section \ref{sec:alg}, and then quickly prove its correctness (that it returns a valid integral cut).
The majority of the proof is devoted to bounding the size of that cut.
This size analysis is organized in a top-down fashion: the first few sections contain one or two gray-boxed lemmas, which are essentially summaries of the rest of the argument to follow, which hide some important definitions or machinery.
On the first read, we encourage the reader to take the gray-boxed lemmas as a black box, and not read further into the proof until comfortable with the earlier parts of the argument.
The gray-boxed lemmas can then later be unpacked by proceeding into the next subsection.




\section{Proof of Main Results \label{sec:body}}

\subsection{Setup, Notation, and Formal Statement of Main Results \label{sec:formalisms}}

Let us first state our main theorems (Theorems \ref{thm:main} and \ref{thm:mainweight}) more formally, so that it is clear what we intend to prove, and introduce some relevant notation.

\paragraph{Statement of Edge Theorems.} We start with the edge versions of these theorems.
Let $G = (V, E, \cost, w)$ be any $n$-node directed graph, where $\cost, w$ hold edge costs and weights, respectively.
Then these theorems state that there is a subset of edges $X \subseteq E$ such that, for every pair of nodes $(s, t)$ with $\dist_w(s, t) \ge 1$,\footnote{Instead of explicitly considering a set of demand pairs $P$, in order to reduce notation, we simply require that $X$ is an integral cut for every pair that is fractionally cut by $w$ -- that is, $P$ is implicitly defined as the demand pairs for which $\dist_w(s, t) \ge 1$.  This clearly implies that $X$ is a cut for any set of demand pairs $P$ that may be of interest.}
every $s \leadsto t$ path in $G$ contains at least one edge in $X$.
Moreover, we have
$$\cost(X) \le \langle \cost, w \rangle \cdot \min\left\{ \Ohat\left(n^{1/3}\right), \Ohat\left(w(E)^{1/2}\right)\right\},$$
where:
$$\cost(X) := \sum \limits_{e \in X} \cost(e), \quad w(E) := \sum \limits_{e \in E} w(e), \quad \langle \cost, w \rangle := \sum \limits_{e \in E} \cost(e) \cdot w(e).$$
We will use similar notation throughout the paper.
We will not consider the edge version directly in our main arguments to follow.
Rather, by Theorem \ref{thm:edgetovertintro} (stated and proved more formally in Theorem \ref{thm:edgetovert}), this claim reduces to our theorems in the vertex setting, which we formally state next.

\paragraph{Statement of Vertex Theorems.} Let $G = (V, E, \cost, w)$ be any $n$-node directed graph, where $\cost, w$ hold vertex costs and weights, respectively.
For a pair of nodes $(s, t)$, we write the ``vertex distance'' between them as
$$\vdist_w(s, t) := \min \limits_{\pi \text{ is an $s \leadsto t$ path}} w\left(\pi \setminus \{s, t\}\right),$$
that is, vertex distance only counts weights of \emph{interior} vertices on $s \leadsto t$ paths and not $(s, t)$ themselves.
In some parts of the paper we will also consider vertex-unweighted distance, still written $\vdist$, which is equivalent to the setting where all vertices have weight $1$.

Our theorems state that we may compute a subset of nodes $X \subseteq V$ such that, for every pair of nodes $(s, t)$ with $\vdist_w(s, t) \ge 1$, every $s \leadsto t$ path $\pi$ in $G$ contains at least one \emph{internal} node in $X$.
That is, we may or may not also have $s, t \in X$, but these nodes do not suffice to cut the pair $(s, t)$.
Moreover, the theorems guarantee the size bound
$$\cost(X) \le \langle \cost, w \rangle \cdot \min \left\{ \Ohat\left(n^{1/3}\right), \Ohat\left(w(V)^{1/2}\right) \right\}.$$
Again, we will not directly consider this statement in our main arguments to follow.
Rather: the bound of $\Ohat(n^{1/3})$ is \emph{implied} by the bound of $\Ohat(w(V)^{1/2})$ via Theorem \ref{thm:introwton} (proved in Theorem \ref{thm:vfgself}), so we will only consider the latter bound.
Additionally, via Theorem \ref{thm:utwreductions} (proved via a sequence of reductions in Section \ref{sec:reductions}), it suffices to prove the bound of $\Ohat(w(V)^{1/2})$ in the special case of uncapacitated, uniform-weight input graphs.
It will be convenient to rephrase the theorem a bit in light of these simplifications; we do this next.

\paragraph{Rephrasing of Uncapacitated, Uniform-Weight Vertex Setting.}

The technical argument in the rest of this section will prove the following lemma.

\begin{mdframed}[backgroundcolor=gray!10]
\begin{lemma} \label{lem:vcr}
Let $G = (V, E)$ be an $n$-node directed graph and let $1 \le L \le n$.
Then in randomized polynomial time, one can construct a set of nodes $X \subseteq V$ that is an integral vertex cut for all node pairs $(a, b)$ at vertex-unweighted distance $\vdist_G(a, b) \ge L$,
and where with high probability $X$ has size
$$|X| \le \Ohat\left(\frac{n}{L}\right)^{3/2}.$$
\end{lemma}
\end{mdframed}

This lemma is morally a rephrasing of the $W$-parametrized vertex theorem in the uncapacitated, uniform-weight setting, as shown in the following proof:

\begin{proof} [Proof of Theorem \ref{thm:mainweight} ($W$-parametrized bound), assuming Lemma \ref{lem:vcr}]
Let $G = (V, E, \cost, w)$ be an input to the vertex flow-cut gap problem.
By Theorem \ref{thm:utwreductions}, we may assume without loss of generality that all nodes $v \in V$ have $\cost(v)=1$ and $w(v) = w^*$, for some value $w^* \le 1$.
Thus, every pair of nodes $(a, b)$ that has vertex-weighted distance $\vdist_{w}(a, b) \ge 1$ must have vertex-unweighted distance $\vdist(a, b) \ge \frac{1}{w^*}$.
By Lemma \ref{lem:vcr}, we may compute an integral cut $X$ for all such pairs, which has size
\begin{align*}
\cost(X) = |X| &\le \Ohat\left(n w^{*}\right)^{3/2}\\
&= \Ohat\left(w(V)^{3/2}\right) \tag*{since $w(V) = nw^*$}\\
&= \langle \cost, w \rangle \cdot \Ohat\left(w(V)^{1/2}\right) \tag*{since costs are $1$. \qedhere}
\end{align*}
\end{proof}

It now remains to prove Lemma \ref{lem:vcr}.
We note that we may assume without loss of generality that $L \ge n^{1/3}$, since otherwise Lemma \ref{lem:vcr} claims a bound of $|X| \le \Ohat(n)$, which holds trivially.
This parameter restriction will be useful at some much-later point in the proof.
From this point onwards, things start to get technical, so we will give some informal intuition before launching into the main argument.

\subsection{Intuition for the Proof of Lemma \ref{lem:vcr} \label{sec:intuition}}








We view our argument as an extension of Gupta's approach \cite{Gupta03}, so we will start by overviewing his argument.
Although his original paper focused on the edge setting and on general weights/costs, we will describe a tweaked version of his algorithm and analysis that removes a few technical details, and which applies only to vertex/uniform weight/uncosted setting:

\begin{theorem} [Implicit in \cite{Gupta03}] \label{thm:guptalem}
Lemma \ref{lem:vcr} holds with a weaker bound of
$|X| \le O\left(\frac{n}{L}\right)^2.$
\end{theorem}

\SetKwInOut{Input}{Input}
\SetKwInOut{Output}{Output}
\SetKwBlock{Loop}{loop}{end}
\DontPrintSemicolon
\begin{mdframed}[backgroundcolor=blue!10]
\begin{algorithm}[H]
\caption{A tweaked version of Gupta's Algorithm \cite{Gupta03} \label{alg:gupta}}

\Input{directed graph $G=(V,E)$, parameter $L \ge 1$}
\Output{integral cut $X \subseteq V$ for all pairs at vertex-unweighted distance $\ge L$}

~

$X \gets \emptyset$\;
$P \gets \left\{ (s, t) \in V\times V \ \mid \ \vdist_G(s, t) \ge L\right\}$\;

~

\While{$P$ is nonempty}{

    $(s, t) \gets$ remove and return a random demand pair from $P$\;
    sample $d$ uniformly at random from interval $[0, L]$\;
    add all nodes $v$ to $X$ for which $d \in \big[\vdist_{G \setminus(X \setminus \{s, t\})}(s, v), \vdist_{G \setminus(X \setminus \{s, t\})}(s, v) + 1\big]$\;
}

\Return{$X$}

\end{algorithm}
\end{mdframed}

\begin{proof} [Proof Sketch for Theorem \ref{thm:guptalem}]
In each round of the main loop of Algorithm \ref{alg:gupta}, we select a demand pair $(s, t)$.
For each node $v$ that is between $(s, t)$ (meaning that an $s \leadsto v \leadsto t$ path exists in the current graph $G \setminus (X \setminus \{s, t\})$), the probability that we cut $v$ is at most $\frac{1}{L}$.
Thus, our goal is to prove that the typical node $v$ is only eligible to be cut in $O(n/L)$ rounds; that is, $v$ is only contained between $O(n/L)$ demand pairs at the time they are selected. 
If we can do so, the expected size of the cut $|X|$ will be
\begin{align*}
\mathbb{E}[|X|] &\le \underbrace{n}_{\text{number of nodes}} \cdot \underbrace{\frac{1}{L}}_{\substack{\text{prob of cutting node each}\\\text{time it is between selected $(s, t)$}}} \cdot \underbrace{O\left(\frac{n}{L}\right)}_{\substack{\text{number of rounds in which}\\\text{node is between selected $(s, t)$}}}\\
&\le O\left(\frac{n^2}{L^2}\right).
\end{align*}
Gupta's key insight is to analyze the \emph{pairwise} connectivity of the nodes in $V$ in order to bound the number of rounds in which each node $v$ lies between the selected demand pair $(s, t)$.
We will either have $\vdist(s, v) \ge \Omega(L)$ or $\vdist(v, t) \ge \Omega(L)$.
Thus, with constant probability, there will be $\Omega(L)$ nodes $u$ for which we \emph{disconnect} the pair $(u, v)$ in this round: that is, there are $\Omega(L)$ nodes $u$ for which there was previously a $u \leadsto v$ path (or $v \leadsto u$ path) in $G \setminus X$, but this is no longer true after adding new nodes to $X$ in this round.
Thus we expect $v$ to be between the selected demand pair $(s, t)$ for only $O(n/L)$ rounds.
\end{proof}

We now explain the key ways in which our argument departs from this one.

\paragraph{Triangle Analysis.}

The core idea behind our improvement is straightforward: instead of analyzing \emph{pairwise} connectivity among nodes in $V$, as in Gupta's analysis, we will instead focus on the \emph{triple-wise} connectivity among nodes.

Very informally, the idea is to show if the expected cut size in a round of the algorithm is large, then there exist special ``triangle structures."
As in the above Gupta-style analysis, these structures allow us to use the fact that the typical node $v$ between some demand pair $(s,t)$ has at least $\Omega(L)$ far-away nodes $u$, where $(u, v)$ lies between the selected demand pair.
However, they also give us an additional advantage: they will imply that some of these pairs $(u, v)$ lies together between \emph{many different} demand pairs.
Thus we are likely to select one such demand pair in an early round of the algorithm, and so $(u, v)$ will be disconnected in even fewer rounds.
The details regarding the new triangle structures that we use in our analysis can be found in Section \ref{sec:shortcharge}, in particular, Lemma \ref{lem:pathstructure}.



There are two serious technical barriers that we encounter along the way to setting up this triangle structure and analysis.
These induce some technical complexities in our algorithm and analysis that depart from Gupta's argument.
We overview these next.

\paragraph{Close-Pair Separation.} A pleasantly simple feature of Gupta's analysis is that it suffices to focus on ``far-apart'' node pairs.
That is, his argument sketched above is roughly based on the following observation: for a reachable pair of nodes $(u, v)$ that are $\Omega(L)$ nodes apart, after only a few demand pairs are considered that have $(u, v)$ between them, the cuts for those demand pairs will probably have separated $(u, v)$.

When focusing on triangles, it turns out that we need a more nuanced argument.
We will roughly argue that either (1) there are useful triangle structures, which force certain far-apart pairs $(u, v)$ to lie between \emph{many} demand pairs (relative to Gupta's argument), or (2) there is a reachable pair of nodes $(u, v)$ that are relatively close together, but which lie between \emph{even more} demand pairs.
Thus, in the second case, $(u, v)$ is less likely to be separated by the cut for any \emph{particular} demand pair, but this is offset by the fact that we are much more likely to select a demand pair that \emph{could} cut $(u, v)$.
As a technical consequence, unlike Gupta, our probabilistic separation guarantees are over both the random cut \emph{and} the random choice of demand pair in each round.
In fact, Gupta's original analysis allowed for an adversarial ordering of the demand pairs, whereas ours truly requires that the demand pairs are considered in random order.

Formalizing the second case requires two new technical details.
One is a function called ``$\val(u, v)$'' that roughly captures the total contribution, across all demand pairs, to the probability of separating $(u, v)$ within the next round(s) of the algorithm (defined in Section \ref{sec:val}).
We formally prove a sort of concentration bound in Lemma \ref{lem:valub}, showing that all pairs of high value will probably have been separated within a small number of rounds.
This ingredient also induces a new technical detail in the algorithm that departs from Gupta's approach.
That is: it is possible that, when we add nodes to $X$ in a round of the algorithm, the value of $\vdist_{G \setminus X}(u, v)$ increases significantly, inducing a sudden increase in $\val(u, v)$.
This could lead to a failure of our concentration bound, if $\val(u, v)$ is small for many rounds and so escapes separation, but then suddenly jumps to a high value in a later round.
To address this, we divide our algorithm into a small number of ``epochs.''
The cuts performed within each epoch depend only on a set of fractional weights $\{w_{s, t}\}$ computed at the \emph{beginning} of an epoch, regardless of which nodes get added to the integral cut $X$ \emph{during} the epoch.
Since these weights do not change throughout the epoch, $\val(u, v)$ does not jump as described previously.
Our concentration inequality then holds only within each epoch.

The second new technical detail is a ``charging scheme,'' which carefully modifies the input instance while charging these modifications to the $\val$ function.
We give details in Section \ref{sec:charging}.
The scheme either identifies a pair $(u,v)$ with high $\val(u, v)$, or allows us to reduce to the setting where our desired triangle structures exist.

\paragraph{Path Witnessing.} In order to argue that there are ``many triangles worth of paths'' among the demand pairs, we first need to argue that the nodes between each demand pair $(s, t)$ can be decomposed into many long internally-node-disjoint paths, which form our triangles.
We also want the number of such paths to scale with the total weights $\{w_{s, t}\}$ in our fractional cuts, since this is the expected cost of the random cut for $(s, t)$.
We call such a decomposition a \emph{path witness}.

A natural attempt to guarantee the existence of path witnesses is the following.
In each round, let $w'_{s, t}$ be a \emph{current} minimum fractional cut on $G \setminus X$.
If the cuts $\{w'_{s, t}\}$ are much cheaper in total than the cuts $\{w_{s, t}\}$ on the remaining demand pairs, this would trigger the end of an epoch, and we would switch to using fractional cut weights $\{w'_{s, t}\}$.
Otherwise, we can apply the MCMF Theorem to say that each demand pair $(s, t) \in P$ has internally-node-disjoint $s \leadsto t$ paths of size $w'_{s, t}(V \setminus X)$, which is approximately the same size as $w_{s, t}(V \setminus X)$, giving our desired path witness.
The reason this attempt fails is that, for technical reasons, in order for our triangles to be useful we will also need to avoid overly heavy nodes.
In particular, we will require that all nodes $v \in V \setminus X$ have weight $w_{s, t}(v) \le \Ohat(1/L)$.
We can no longer apply the MCMF Theorem for fractional min cuts $w'_{s, t}$ with upper-bounded weights, so we do not get our desired path witnesses.

Our solution, which is stated and applied in Section \ref{sec:charging} (Lemma \ref{lem:pathwitness}) and then proved in Section \ref{sec:pathwitness}, is a proof that this MCMF application can be recovered if: (1) we allow the paths in our witness to be in the \emph{transitive closure} of $G \setminus X$, rather than in $G \setminus X$ directly, and (2) we relax the weight upper bound on $w'_{s, t}$ by a $\cdot O(1)$ factor relative to the weight upper bound used for $w_{s, t}$.
These constraints lead to some delicate parameter balancing in our argument, but give us a workable approximate path witness.

\subsection{Main Algorithm \label{sec:alg}}

Our proof of Lemma \ref{lem:vcr} is based on an analysis of the following algorithm.

\SetKwInOut{Input}{Input}
\SetKwInOut{Output}{Output}
\SetKwBlock{Loop}{loop}{end}
\DontPrintSemicolon
\begin{mdframed}[backgroundcolor=blue!10]
\begin{algorithm}[H]
\caption{$\mvc(G, L)$}
\label{alg:main-alg}

\Input{directed graph $G=(V,E)$, parameter $L \ge 1$}
\Output{cut $X \subseteq V$ for pairs at vertex-unweighted distance $\ge L$}

\begin{mdframed}[backgroundcolor=blue!20, linecolor=blue!20, rightmargin=25pt]
\tcp{Setup}
\vskip3pt
$X \gets \emptyset$\;
$P \gets \left\{ (s, t) \in V\times V \ \mid \ \vdist_G(s, t) \ge L\right\}$\;
$\epoch \gets 1$

\ForEach{pair of nodes $(s, t) \in P$}{
    $w_{s, t} \gets $ function that assigns all nodes weight $\frac{1}{L}$\;
}
\end{mdframed}

\While{$P$ is nonempty}{

    \begin{mdframed}[backgroundcolor=blue!20, linecolor=blue!20, rightmargin=40pt]
    \tcp{Compute new candidate weights}
    \vskip3pt
    \ForEach{$(s, t) \in P$}{
        $w'_{s, t} \gets$ minimum fractional $(s, t)$ cut in which all nodes in $X$ have weight $1$, and all nodes in $G \setminus X$ have weight $\le \frac{4^{\epoch}}{L}$
    }
    \end{mdframed}

    \begin{mdframed}[backgroundcolor=blue!20, linecolor=blue!20, rightmargin=40pt]
    \tcp{If candidate weights are much cheaper, start new epoch}

    \vskip 5pt
    \If{$\mass\left( \{w_{s, t}\} \right) > \log n \cdot \mass\left( \{w'_{s, t}\}\right)$}{
    \vskip2pt
        $\epoch \gets \epoch + 1$\;
        $\{w_{s, t}\} \gets \{w'_{s, t}\}$\;
    }
    \end{mdframed}    
    
    \begin{mdframed}[backgroundcolor=blue!20, linecolor=blue!20, rightmargin=40pt]
    \tcp{Take a "random level cut" with respect to current weights}
    \vskip3pt
    $(s, t) \gets$ remove and return a random demand pair from $P$\;
    sample $d$ uniformly at random from interval $[0, 1]$\;
    add all nodes $v$ to $X$ for which $d \in \big[\vdist_{s, t}(s, v), \vdist_{s, t}(s, v) + w_{s, t}(v)\big]$\;
    \end{mdframed}

}

\Return{$X$}
\end{algorithm}
\end{mdframed}

This algorithm introduces two pieces of notation that we will continue to use through the rest of the analysis.
First, we write $\vdist_{s, t}$ as a shorthand for the vertex-weighted distance under the current weight function $w_{s, t}$.
Second, we define:
\begin{definition} [$\mass$ function]
At any point in the algorithm, we define
$$\mass\left(\{w_{s, t}\}\right) := \sum \limits_{(s, t) \in P} \sum \limits_{v \in V \setminus X} w_{s, t}(v).$$
\end{definition}

Our goal in the rest of this section is to prove correctness of the algorithm, i.e., that $X$ is a valid integral cut.


\begin{lemma} \label{lem:rlccorrect}
When we select a pair $(s, t) \in P$ in the main loop of the algorithm, if for two nodes $u, v \in V$ we sample $d$ in the interval $\left[\vdist_{s, t}(s, u), \vdist_{s, t}(s, v)\right),$
then afterwards we have either $u \in X$ or $X$ is an integral cut for $(u, v)$.
\end{lemma}
\begin{proof}
Let $\pi$ be an arbitrary $u \leadsto v$ path, and let its vertex sequence be
$\pi = \left(u=x_0, x_1, \dots, x_{k-1}, x_k=v\right).$
By assumption we have
$\vdist_{s, t}(s, u) \le d \le \vdist_{s, t}(s, v).$
Thus, by an intermediate value type argument, there is some index $0 \le i \le k-1$ for which we have
$\vdist_{s, t}(s, x_i) \le d \le \vdist_{s, t}(s, x_{i+1}).$
Additionally, by the triangle inequality over vertex distances, we have
$\vdist_{s, t}(s, x_{i+1}) \le \vdist_{s, t}(s, x_i) + w_{s, t}(x_i).$
Combining this with the previous inequality, we have
$\vdist_{s, t}(s, x_i) \le d \le \vdist_{s, t}(s, x_i) + w_{s, t}(x_i).$
Thus we add $x_i$ to $X$ in the random level cut.
So either $x_i=u$ or $X$ contains an internal vertex of $\pi$; in either case, the lemma holds.
\end{proof}

\begin{lemma}
Algorithm \ref{alg:main-alg} returns a correct integral cut $X$ (deterministically).
\end{lemma}
\begin{proof}
Each pair $(s, t) \in P$ will be considered in some round of the algorithm.
When we do, consider any $s \leadsto t$ path $\pi$, and let $x$ be its second node.
We have $\vdist_{s, t}(s, x) = 0$ (since they are connected by a single edge), and we also have $\vdist_{s, t}(s, t) \ge 1$ (since $w_{s, t}$ is a fractional cut).
Thus we will necessarily sample $d$ in the interval $[\vdist_{s, t}(s, x), \vdist_{s, t}(s, t)]$.
By the previous lemma, this implies that we will have $x \in X$ or $X$ is an integral cut for $(x, t)$.
In either case, $X$ contains an internal vertex of $\pi$.
Since an arbitrary $s \leadsto t$ path $\pi$ has an internal vertex in $X$, this means $X$ forms a cut for $(s, t)$, completing the lemma.
\end{proof}

We now turn our attention to bounding $|X|$, the size of the cut returned by this algorithm.
Our goal will be to show that the algorithm succeeds with a large constant probability; naturally, one can boost this to a high probability guarantee by re-running the algorithm $O(\log n)$ times and taking the minimum $|X|$ among the cuts returned.

\subsection{Bounding Cut Size}

Let us say that an \emph{epoch} of the algorithm is a sequence of iterations of the main \emph{while} loop in which the $\epoch$ parameter is not increased (i.e., the final \emph{if} condition does not trigger).
Within an epoch, each iteration of the main loop in which one demand pair is removed and cut is called a \emph{round}.
We can control the total number of epochs with a simple counting argument:

\begin{lemma} [Bound on Number of Epochs] \label{lem:numepochs}
The number of epochs in the algorithm is at most $O\left(\frac{\log n}{\log \log n}\right)$.
\end{lemma}
\begin{proof}
We will use the quantity $\mass(\{w_{s, t}\})$ as a potential function.
Initially, we have the (somewhat loose) bound
$$\mass\left(\{w_{s, t}\}\right) \le \underbrace{n^2}_{\text{node pairs } |P|} \cdot \underbrace{n}_{\text{nodes}} \cdot \underbrace{\frac{1}{L}}_{\text{weight per node}} \le n^{3}.$$
If there is at least one demand pair $(s, t) \in P$ for which $X$ is not yet a cut, we have
$$\mass\left(\{w_{s, t}\}\right) \ge 1,$$
since the weights in $w_{s, t}$ on $V \setminus X$ must fractionally cut at least one path, so they must sum to at least $1$.
On the other hand, if $X$ is a cut for every demand pair in $P$, then we will trivially have
$$\mass\left(\{w_{s, t}\}\right) = 0$$
since the minimizing choice of fractional cuts $\{w'_{s, t}\}$ in the algorithm does not need to assign any weight to nodes in $V \setminus X$.
Finally, each time we start a new epoch, by construction $\mass(\{w_{s, t}\})$ decreases by at least a factor of $\log n$.
It follows that all demand pairs $P$ will be cut by $X$ within number of epochs
\begin{align*}
O\left(\log_{(\log n)} n^3\right) = O\left(\frac{\log n}{\log \log n}\right). \tag*{\qedhere}
\end{align*}
\end{proof}

In part, the significance of the previous lemma is that for any vertex $v \in V \setminus X$ and any weight function $w_{s, t}$, at any point in the algorithm we will have
$$w_{s, t}(v) \le \frac{4^{\epoch}}{L} \le \frac{4^{O(\log n / \log \log n)}}{L} = \frac{n^{o(1)}}{L} = \Ohat\left(\frac{1}{L}\right),$$
meaning that weights do not blow up too far beyond $\frac{1}{L}$ throughout the algorithm.
We remark that weight blowup is the only reason we lose $n^{o(1)}$ factors, rather than merely $\texttt{polylog} n$ factors, in our bounds.

\begin{lemma} [Cost of $i$-Round Epoch] \label{lem:epochcost}
For any epoch and any positive integer $i$, the expected number of new nodes added to $X$ during the first $i$ rounds of the epoch is at most
$$\mass\left(\{w_{s, t}\}\right) \cdot \left( \frac{i}{|P|}\right)$$
where $\mass, |P|$ are measured at the beginning of the epoch.
\end{lemma}
\begin{proof}
The probability that any given pair $(s, t) \in P$ is selected in the first $i$ rounds of the epoch is at most\footnote{The probability is generally a bit less than this, only because there is a chance that the epoch ends in fewer than $i$ rounds.} $i / |P|$.
If $(s, t)$ is selected, then the expected number of new nodes added to $X$ in that round of the \emph{while} loop is:
\begin{align*}
     & \ \sum_{v \in V \setminus X} \Pr\left[\text{sample } d \in \left[\vdist_{s, t}(s, v), \vdist_{s, t}(s, v) + w_{s, t}(v)\right]\right]\\
    \le& \ \sum_{v \in V \setminus X} w_{s, t}(v).
\end{align*}
The lemma now follows by applying linearity of expectation to sum this bound over all demand pairs.
\end{proof}

Finally, we will need the following lemma, bounding the number of rounds per epoch.
The proof is far more involved, and so we will begin it in the next subsection.

\begin{mdframed}[backgroundcolor=gray!10]
\begin{lemma} [Bound on Rounds per Epoch] \label{lem:roundbound}
With high probability, every epoch will terminate within at most
$$\Ohat\left(\frac{|P| \cdot n^{3/2}}{\mass\left(\{w_{s, t}\}\right) \cdot L^{3/2}}\right)$$
rounds of the main loop, where $\mass$ and $|P|$ are measured at the beginning of the epoch.
\end{lemma}
\end{mdframed}

We next put the pieces together, and see how these lemmas imply Lemma \ref{lem:vcr}.

\begin{proof} [Proof of Lemma \ref{lem:vcr}, assuming Lemma \ref{lem:roundbound}]
We assume that the high probability event in Lemma \ref{lem:roundbound} holds in every epoch.
We have:
\begin{align*}
\mathbb{E}[|X|] &\le
\underbrace{\Ohat(1)}_{\text{number of epochs (Lemma \ref{lem:numepochs})}} \cdot \underbrace{\frac{\mass\left(\{w_{s, t}\}\right)}{|P|}}_{\text{Lemma \ref{lem:epochcost}}} \cdot \underbrace{\Ohat\left(\frac{|P| \cdot n^{3/2}}{\mass\left(\{w_{s, t}\}\right) \cdot L^{3/2}}\right)}_{\text{Lemma \ref{lem:roundbound}}}\\
&\le \Ohat\left(\frac{n^{3/2}}{L^{3/2}}\right). \tag*{\qedhere}
\end{align*}
\end{proof}

It now remains to prove Lemma \ref{lem:roundbound}.

\subsection{Proof of Lemma \ref{lem:roundbound} and the $\val$ Function \label{sec:val}}

The following $\val$ function will be used in the proof of Lemma \ref{lem:roundbound}:

\begin{definition} [$\val$ function]
    In an epoch, for a pair of nodes $(u, v)$ and a demand pair $(s, t) \in P$, we define
    $$\val_{s, t}(u, v) := \Pr_{d \sim [0, 1]}\bigg[ d \in \left[\vdist_{s, t}(s, u), \vdist_{s, t}(s, v)\right]\bigg]$$
    and
    $$\val(u, v) := \sum \limits_{(s, t) \in P} \val_{s, t}(u, v).$$
\end{definition}

The motivation for this definition is from Lemma \ref{lem:rlccorrect}: $\val_{s, t}(u, v)$ is precisely the probability that a random level cut for $(s, t)$ falls in the proper interval to either add $u$ to $X$, or to cut the pair $(u, v)$.
Thus, intuitively, a pair with high $\val(u, v)$ is one where it is relatively likely that, when we next choose a random $(s, t)$ and then take a random level cut for $(s, t)$, we will add $u$ to $X$ or cut the pair $(u, v)$.
This interpretation is the motivation behind the following lemma:

\begin{lemma} \label{lem:valub}
With high probability,\footnote{As usual, ``with high probability'' means probability at least $1 - \frac{1}{n^C}$, where the constant $C$ may be pushed arbitrarily high by choice of high enough implicit constants.} for all pairs of nodes $(u, v)$ with $\val(u,v) > 0$, if 
$$\Omega\left( \frac{|P|}{\val(u, v)} \cdot \log n \right)$$
rounds of the epoch pass then there will no longer be a $u \leadsto v$ path in $G \setminus X$.
(This is either because $X$ is a cut for $(u, v)$, or because $u \in X$.
The value of $|P|$ is measured at the start of the epoch.)
\end{lemma}
\begin{proof}
Let $(u,v)$ be a pair of nodes with positive $\val(u, v)$.
In any round where we consider a demand pair $(s,t)$, if we sample
$$d \in \left[\vdist_{s, t}(s, u), \vdist_{s, t}(s, v)\right]$$
then by Lemma \ref{lem:rlccorrect} there will be no more $u \leadsto v$ paths in $G \setminus X$.
By definition, the probability that this sampling occurs is $\val_{s, t}(u, v)$.

For any positive integer $i$, assuming that the epoch lasts at least $i$ rounds, the subset of pairs $P' \subseteq P$ selected in the first $i$ rounds is a uniform $i$-element subset of $P$.
The set of all such subsets is denoted $P \choose i$.
Using this, we may bound:
\begin{align*}
\Pr \left[ \text{exists } u \leadsto v \text{ path in } G \setminus X \text{ after $i$ rounds} \right] &\le \frac{1}{\binom{|P|}{i}}\sum \limits_{P' \in \binom{P}{i}} \prod \limits_{(s, t) \in P'} (1 - \val_{s, t}(u, v))\\
&\le  \left( \frac{1}{|P|} \cdot \sum \limits_{(s, t) \in P} 1 - \val_{s, t}(u, v) \right)^i \tag*{Maclaurin's inequality}\\
&= \left(1 - \frac{\val(u, v)}{|P|}\right)^i\\
&\le e^{-i \cdot \val(u, v) / |P|} \tag*{since $1-x \le e^{-x}$ for all $x$.}
\end{align*}

Hence, when $i \ge C \ln n |P| / \val(u, v)$, the probability of a $u \leadsto v$ path remaining in $G \setminus X$ is at most $n^{-C}$, completing the proof.
\end{proof}

The following lemma is more involved; we will, in turn, begin its proof in the following section.
\begin{mdframed}[backgroundcolor=gray!10]
\begin{lemma} [Lower Bound on $\val$] \label{lem:vallb}
In any round of an epoch, there exists a pair of nodes $(u, v)$ with
$$\val(u, v) \ge \mass\left(\{w_{s, t}\}\right) \cdot \Omegahat\left(\frac{L}{n}\right)^{3/2}.$$
and where there exists a $u \leadsto v$ path in $G \setminus X$.
\end{lemma}
\end{mdframed}

Using these lemmas, we can now prove Lemma \ref{lem:roundbound}:
\begin{proof} [Proof of Lemma \ref{lem:roundbound}, assuming Lemma \ref{lem:vallb}]
Assume that the high-probability event in Lemma \ref{lem:valub} holds, and let $(u, v)$ be a pair of nodes satisfying Lemma \ref{lem:vallb}.
Since Lemma \ref{lem:vallb} guarantees that there is a $u \leadsto v$ path in $G \setminus X$, it follows from Lemma \ref{lem:valub} that the number of rounds that has passed is at most
\begin{align*}
&\Ohat\left(\frac{|P|}{\val(u, v)}\right)\\
\le& \  \Ohat\left(\frac{|P|}{\mass\left(\{w_{s, t}\}\right) \cdot \frac{L^{3/2}}{n^{3/2}}}\right) \tag*{Lemma \ref{lem:vallb}}\\
=& \  \Ohat\left(\frac{|P| \cdot n^{3/2}}{\mass\left(\{w_{s, t}\}\right) \cdot L^{3/2}}\right). \tag*{\qedhere}
\end{align*}
\end{proof}

\subsection{Proof of Lemma \ref{lem:vallb} and Charging Scheme \label{sec:charging}}

The following lemma allows us to find a path system $R = (V, \Pi)$ in each round of the algorithm, which ``witnesses'' the current fractional cuts $\{w_{s, t}\}$, in a few useful ways.

\begin{mdframed}[backgroundcolor=gray!10]
\begin{lemma} [Path System Witness Lemma]\label{lem:pathwitness}
At the beginning of each round, there exists a path system $R = (V, \Pi)$ (over the same vertex set as $G$), as well as a partition of the paths $\Pi$ into (possibly empty) parts
$$\Pi = \bigcup \limits_{(s, t) \in P} \Pi_{s, t},$$
with all of the following properties:

\begin{enumerate}
\item For each part $\Pi_{s, t}$, the paths in $\Pi_{s, t}$ are pairwise vertex-disjoint.

\item For any path $\pi$ in a part $\Pi_{s, t}$, $\pi$ is a not-necessarily-contiguous vertex subsequence of some $s \leadsto t$ path in $G \setminus X$.
Moreover, letting the vertex sequence be $\pi = (v_0, \dots, v_k)$, we have $\frac{L}{\lambda} \le k \le L$ for some $\lambda = n^{o(1)}$, and for each index $i$ we have
$$\vdist_{s, t}(s, v_i) + \frac{1}{L} \le \vdist_{s, t}(s, v_{i+1}),$$
and also $\vdist_{s, t}(s, v_k) \le 1$.

\item $|\Pi| \ge \Omegahat\left(\mass\left(\{w_{s, t}\}\right)\right)$
\end{enumerate}
\end{lemma}
\end{mdframed}

The proof of this lemma is technical, and is deferred to Section \ref{sec:pathwitness}.
We note the role of the new parameter $\lambda = n^{o(1)}$, since we will track some $\lambda$ dependencies in the following argument.
We do not care about optimizing these factors, and will often hide them in $\Ohat(\cdot)$ notation.
This is just our means of navigating some dependencies in the various $n^{o(1)}$ factors that arise in the following lemmas, and emphasizing that the various $n^{o(1)}$ factors can be selected in a non-circular fashion. 
In the following lemmas we will write $R = (V, \Pi)$ for a path system from this lemma.

Our next step is to set up a \emph{charging scheme}, which in each round of the algorithm, will either identify a pair of nodes that carries high value or reduce to a path witness $R$ with a more convenient structural form.
At a high level: we will define a relatively small set of node pairs $S$, and then we will carefully \emph{charge} value to the pairs in $S$.
One case of our proof will work by arguing that if the charging scheme runs for many rounds, then one of the pairs $(u, v) \in S$ has been charged for a lot of value, implying that $\val(u, v)$ is high enough to satisfy Lemma \ref{lem:vallb}.
The other case will argue that if the charging scheme terminates early, then there exists a useful combinatorial structure in $R$.

\begin{definition} [Path Prefixes, Suffixes, Degrees]
The prefix, suffix of a path $\pi \in \Pi$ are the contiguous subpaths respectively consisting of the first, last $\lfloor |\pi|/4 \rfloor$ nodes of $\pi$.
For a node $v$, we write $\texttt{deg}_{\Pi}(v)$ for the number of paths in $\Pi$ that contain $v$, and we write $\texttt{suffdeg}_{\Pi}(v)$ for the number of paths in $\Pi$ that contain $v$ in their suffix.
\end{definition}

\begin{lemma} [Base Path] \label{lem:basepath}
Let $R$ be a path system as in Lemma \ref{lem:pathwitness}, and let $d$ be the average of $\texttt{deg}_{\Pi}(v)$ over its nodes.
Then there exists a path $\pi_b \in \Pi$ with the property that
$$\sum \limits_{v \in \pi_b} \texttt{suffdeg}_{\Pi}(v) \ge \frac{L d}{\lambda}.$$
\end{lemma}
\begin{proof}
Sample a path $\pi_b \in \Pi$ uniformly at random.
Note that each node $v$ contributes to the left-hand sum in the lemma statement iff $v \in \pi_b$, and if so, it contributes $\texttt{suffdeg}_{\Pi}(v)$ to the sum.
We therefore have
\begin{align*}
\mathbb{E}\left[ \sum \limits_{v \in \pi_b} \texttt{suffdeg}_{\Pi}(v)\right] &= \sum \limits_{\pi \in \Pi} \Pr\left[\pi \text{ selected as } \pi_b\right] \cdot \sum \limits_{v \in \pi} \texttt{suffdeg}_{\Pi}(v)\\
&\ge \frac{1}{|\Pi|} \cdot \sum \limits_{\pi \in \Pi} \sum \limits_{v \in \texttt{suff}(\pi)} \texttt{suffdeg}_{\Pi}(v)\\
&= \frac{1}{|\Pi|} \cdot \sum \limits_{v \in V} \texttt{suffdeg}_{\Pi}(v)^2\\
&\ge \frac{1}{|\Pi|n} \cdot \left( \sum \limits_{v \in V} \texttt{suffdeg}_{\Pi}(v) \right)^2 \tag*{Cauchy-Schwarz}\\
&= \frac{1}{|\Pi|n} \cdot \Theta(nd)^2 \tag*{average \texttt{deg} within $\cdot \Theta(1)$ of average \texttt{suffdeg}}\\
&= \Theta\left(\frac{nd^2}{|\Pi|}\right).
\end{align*}
By Lemma \ref{lem:pathwitness}, all paths in $\Pi$ have $\ge \frac{L}{\lambda}$ nodes, so the average degree is at least $d \ge |\Pi|L / (\lambda n).$
Applying this substitution for one factor of $d$ in the numerator, we conclude that
\begin{align*}
\mathbb{E}\left[ \sum \limits_{v \in \pi_b} \texttt{suffdeg}_{\Pi}(v)\right] \ge \frac{Ld}{\lambda}.\tag*{\qedhere}
\end{align*}
\end{proof}

In the following we fix a particular path $\pi_b \in \Pi$ that satisfies the previous lemma, which we will call the \emph{base path}.

\begin{lemma} \label{lem:supershortcut}
There exists an edge set $S \subseteq V \times V$ of size $|S| \le \Oish(L)$ with the following properties:
\begin{itemize}
\item For every edge $(u, v) \in S$, there is a $u \leadsto v$ path in $G \setminus X$.

\item For any two nodes $x, y \in \pi_b$ for which there exists an $x \leadsto y$ path in $G \setminus X$, we have $\dist_S(x, y) \le 2$.

\item For each node $x \in \pi_b$, there is a set $M_x$ of $|M_x| \le O(\log n)$ nodes such that, for all $y$ with $\dist_S(x, y) = 2$, we specifically have $(x, m), (m, y) \in S$ for some $m \in M_x$.
\end{itemize}
\end{lemma}
\begin{proof}
This is essentially a mild expansion of a standard lemma in the literature on shortcut sets and hopsets (see, e.g., Lemma 1.3 of \cite{Raskhodnikova10}).
We construct $S$ as follows.
Let $m$ be the median node along $\pi_b$ (rounding arbitrarily, if needed).
For each node $v \in \pi_b$, add the edge $(v, m)$ to $S$ if there is a $v \leadsto m$ path in $G \setminus X$, and add the edge $(m, v)$ to $S$ if there is a $m \leadsto v$ path in $G \setminus X$.
Then, recurse on the prefix of $\pi_b$ up to $m$, and also recurse on the suffix of $\pi_b$ following $m$ (neither the prefix nor the suffix used in the recursion include $m$ itself).
The recursion halts when there is $0$ or $1$ node remaining in the path.

The size $|S| =: s$, as a function of the path length $|\pi_b| =: \ell \le O(L)$, is governed by the recurrence
$$s(\ell) \le 2\ell + 2 s(\ell/2),$$
which solves to $s(\ell) \le O(\ell \log \ell) = \Oish(L)$, as claimed.
For the $\dist_S(x, y) \le 2$ guarantee, let us consider two distinct nodes $x, y \in \pi_b$ with $(x, y)$ reachable.
Let $\pi' \subseteq \pi_b$ be the contiguous subpath considered at some level of recursion such that $x, y \in \pi'$, but they lie (weakly) on opposite sides of the median node $m \in \pi'$.
We consider cases:
\begin{itemize}
\item Suppose that the nodes have the order $x \le m \le y$ along $\pi'$.
Then we have $(x, m), (m, y)$ reachability due to the edges in $\pi'$, so we will add $(x, m) \in S$ (unless $x=m$), and also $(m, y) \in S$ (unless $y=m$), witnessing a $\le 2$-edge path from $x$ to $y$.

\item Suppose that the nodes have the order $y \le m \le x$ along $\pi'$.
Since we assume $(x, y)$ reachability, and we have $(y, m)$ reachability due to the edges in $\pi'$, we therefore have $(x, m)$ reachability and so we will add $(x, m) \in S$ (unless $x=m$).
Similarly, we have $(m, x)$ reachability due to the edges in $\pi'$, and we have $(x, y)$ reachability by assumption, and so we have $(m, y)$ reachability and we will add $(m, y) \in S$ (unless $y=m$).
This witnesses a $\le 2$-edge path from $x$ to $y$.
\end{itemize}
In the previous cases, when $x \ne m$ and $y \ne m$ and so we have $\dist_S(x, y)=2$, we have $(x, m), (m, y) \in S$ for a median node $m$ that is in the same subpath as $x$ at some recursion depth.
Since the recursion depth is $O(\log n)$, the third property follows by taking $M_x$ to be the set of all such median nodes.
\end{proof}

In the following we write $S$ for a fixed edge set from this lemma.
We use it as follows:

\begin{mdframed}[backgroundcolor=blue!10]
\textbf{Charging Scheme:}
With respect to a parameter $1 \le \sigma \le L$ that we will choose later, we repeat the following process until no longer possible:
\begin{itemize}
\item Find a path $\pi \in \Pi$ with the following property: there is a node $x \in \texttt{suffix}(\pi) \cap \pi_b$, and an edge $e \in S$ of the form $e=(x, y)$ or $e=(y, x)$, such that
$$\val_{\texttt{endpoints}(\pi)}(e) \ge \frac{\sigma}{\lambda^2 L}.$$

\item Charge $\val_{s, t}(e)$ points of value to the edge $e$.
Then delete $x$ from $\texttt{suffix}(\pi)$.
\end{itemize}
\end{mdframed}

To clarify the scheme: when we delete $x$ from $\texttt{suffix}(\pi)$, we do not recompute $\texttt{suffix}(\pi)$; rather, $\texttt{suffix}(\pi)$ is defined as the surviving subset of the original $\lfloor |\pi|/4 \rfloor$ suffix nodes.
Note that deleting $x$ does not harm any of the properties of Lemma \ref{lem:pathwitness}, in part because only a constant fraction of the nodes in a path lie in the suffix, so path lengths will decrease by only a constant fraction.
In each round of Algorithm \ref{alg:main-alg}, we recompute the path witness $R$ from scratch, and then run this charging scheme on the new system $R$.
This all occurs in the analysis, not the algorithm, so we do not need to track runtime.

The following lemma relates the charging scheme to our lower bound on $\val$:
\begin{lemma} \label{lem:valsplit}
If an edge $e \in S$ is charged for $\alpha$ total points of value throughout the charging scheme, then $\val(e) \ge \frac{\alpha}{2}$.
\end{lemma}
\begin{proof}
Consider any node pair $(x, y) \in S$.
For any fixed demand pair $(s, t) \in P$, in each round where we select a path $\pi \in \Pi_{s, t}$ and then choose to charge $(x, y)$, we will then either delete $x$ from $\pi$ or we will delete $y$ from $\pi$.
By Lemma \ref{lem:pathwitness} the paths in $\Pi_{s, t}$ are pairwise node-disjoint.
Thus we will charge $(x, y)$ at most twice via paths $\pi \in \Pi_{s, t}$, after which neither $x$ nor $y$ participate in paths in $\Pi_{s, t}$, so it will not be charged again.
Hence $\val_{s, t}(x, y)$ is at least half as large as the total value charged to $(x, y)$ via paths in $\Pi_{s, t}$.
The lemma then follows by summing over the demand pairs.
\end{proof}

From this, we can get an easy bound in the case where the charging scheme runs for enough rounds:

\begin{lemma} \label{lem:longcharge}
If the charging scheme runs for $\ge \frac{Ld}{2\lambda}$ rounds, then there is an edge $e \in S$ with
$$\val(e) \ge \Omegahat\left(\frac{\sigma d}{L}\right).$$
\end{lemma}
\begin{proof}
In each round we charge $\Omegahat(\sigma/L)$ value to some edge $e \in S$.
Thus the average edge $e \in S$ is charged for
\begin{align*}
\frac{1}{|S|} \cdot \underbrace{\Omegahat(Ld)}_{\text{number of rounds}} \cdot \underbrace{\Omegahat(\sigma/L)}_{\text{charge per round}} &\ge \Omegahat\left(\frac{\sigma d}{L}\right)
\end{align*}
total value.
The claim then follows by applying Lemma \ref{lem:valsplit} on this edge.
\end{proof}

We now need a different argument in the case where the charging scheme runs for fewer rounds than the bound in the previous lemma.
We will show the following bound:

\begin{mdframed}[backgroundcolor=gray!10]
\begin{lemma} \label{lem:shortcharge}
If the charging scheme halts within $\le \frac{Ld}{2\lambda}$ rounds, then there is an edge $e \in V \times V$ with
$$\val(e) \ge \Omegahat\left(\frac{L^2 d}{\sigma n}\right).$$
\end{lemma}
\end{mdframed}

The proof of this lemma requires a new suite of ideas, exploiting certain combinatorial structures that are guaranteed to exist if and when the charging scheme halts early.
We defer it to Section \ref{sec:shortcharge}.
We now see how it implies Lemma \ref{lem:vallb}:

\begin{proof} [Proof of Lemma \ref{lem:vallb}, assuming Lemmas \ref{lem:pathwitness} and \ref{lem:shortcharge}]
We calculate the choice of $\sigma$ that balances the expressions in Lemmas \ref{lem:longcharge} and \ref{lem:shortcharge} as follows:
\begin{align*}
\frac{d\sigma}{L} &= \frac{L^2 d}{\sigma n}\\
\sigma^2 &= L^3 n^{-1}\\
\sigma &= L^{3/2} n^{-1/2}.
\end{align*}
Recall that we have $n^{1/3} \le L \le n$, and through this parameter restriction we have $1 \le \sigma \le L$, as required by the charging scheme.
We may therefore apply the bound in Lemmas \ref{lem:longcharge} and \ref{lem:shortcharge} (which are equal, under the previous choice of $\sigma$).
These imply that there is a pair of nodes $(u, v)$, with a $u \leadsto v$ path in $G \setminus X$, and where
\begin{align*}
\val(u, v) &\ge \Omegahat\left(\frac{\sigma d}{L}\right) \tag*{Lemma \ref{lem:longcharge}}\\
&= \Omegahat\left(\frac{dL^{1/2}}{n^{1/2}}\right) \tag*{substitute $\sigma$}\\
&\ge |\Pi| \cdot \Omegahat\left(\frac{L}{n}\right)^{3/2} \tag*{since $d \ge \frac{|\Pi|\cdot (L/\lambda)}{n}$}\\
&= \mass\left(\{w_{s, t}\}\right) \cdot \Omegahat\left(\frac{L^{3/2}}{n^{3/2}}\right) \tag*{Lemma \ref{lem:pathwitness}. \qedhere}
\end{align*}
\end{proof}

\subsection{Proof of Lemma \ref{lem:shortcharge} \label{sec:shortcharge}}

The following technical lemma shows that our $\val$ functions satisfy the triangle inequality:
\begin{lemma} [$\val$ Triangle Inequality] \label{lem:valtrineq}
For all $u, v, x, (s, t)$, we have
$\val_{s, t}(u, x) \le \val_{s, t}(u, v) + \val_{s, t}(v, x).$
\end{lemma}
\begin{proof}
In the following calculations we will write $\len$ for interval length, and $\vdist$ as a shorthand for $\vdist_{s, t}$ (suppressing the subscript).
We have:
\begin{align*}
\val_{s, t}(u, x) &\le \len\left([0, 1] \ \bigcap \ [\vdist(s, u), \vdist(s, x)] \right)\\
&\le \len\left([0, 1] \ \bigcap \ \left([\vdist(s, u), \vdist(s, v)] \cup [\vdist(s, v), \vdist(s, x)]\right)\right)\\
&= \len\left(\left([0, 1] \cap [\vdist(s, u), \vdist(s, v)]\right) \ \bigcup \ \left([0, 1] \cap [\vdist(s, v), \vdist(s, x)]\right)\right) \tag*{DeMorgan}\\
&\le \len\left([0, 1] \cap [\vdist(s, u), \vdist(s, v)]\right) \ \mathlarger{+} \ \len\left([0, 1] \cap [\vdist(s, v), \vdist(s, x)]\right)\\
&= \val_{s, t}(u, v) + \val_{s, t}(v, x). \tag*{\qedhere}
\end{align*}
\end{proof}

Using this, we can identify a useful combinatorial structure in our path system:

\begin{lemma} [Path System Structure] \label{lem:pathstructure}
If the charging scheme halts within $\le \frac{Ld}{2\lambda}$ rounds, then once it halts, there exists a node $u \in V$ and a subset of paths $Q \subseteq \Pi$ of size
$$|Q| \ge \Omegahat\left(\frac{L^2d}{\sigma n}\right)$$
where all paths $q \in Q$ have $u \in \texttt{prefix}(q)$, and also $\texttt{suffix}(q)$ intersects $\pi_b$.
\end{lemma}
\begin{proof}
In each round of the charging scheme, we delete one node from one suffix.
Thus, after these deletions, we still have the inequality
$$\sum \limits_{v \in \pi_b} \texttt{suffdeg}_{\Pi}(v) \ge \underbrace{\frac{Ld}{\lambda}}_{\text{initial sum}} - \underbrace{\frac{Ld}{2\lambda}}_{\text{number of deletions}} \ge \Omegahat(Ld).$$
Additionally, we claim that once the charging scheme halts, every path $\pi \in \Pi$ satisfies $|\texttt{suffix}(\pi) \cap \pi_b| \le \sigma$.
To see this, let $\pi \in \Pi$ with $|\texttt{suffix}(\pi) \cap \pi_b| > \sigma$, and let $u, v$ respectively be the first, last nodes along $\texttt{suffix}(\pi)$ that are contained in $\texttt{suffix}(\pi) \cap \pi_b$.
Since $u, v$ are at least $\sigma$ steps apart along $\pi$, by Lemma \ref{lem:pathwitness}, letting $(s, t)$ be the endpoints of $\pi$, we have
$$\val_{s, t}(u, v) \ge \vdist_{s, t}(s, v) - \vdist_{s, t}(s, u) \ge \frac{\sigma}{L}.$$
Then, by Lemma \ref{lem:supershortcut}, we either have $(u, v) \in S$ or we have $(u, m), (m, v) \in S$ for some node $m$.
In the former case, we can successfully charge $\frac{\sigma}{L} \gg \frac{\sigma}{\lambda^2}$ value to $(u, v)$, and then delete $u$ or $v$.
In the latter case, by Lemma \ref{lem:valtrineq}, we have
$$\val_{s, t}(u, m) + \val_{s, t}(m, v) \ge \val_{s, t}(u, v) \ge \frac{\sigma}{L},$$
and so we can charge at least $\frac{\sigma}{2L} \gg \frac{\sigma}{\lambda^2 L}$ value to either $\val_{s, t}(u, m)$ or $\val_{s, t}(m, v)$, and then delete $u$ or $v$.
So the charging scheme will not halt until no more paths with $|\texttt{suffix}(\pi) \cap \pi_b| > \sigma$ exist.

We will use this bound to obtain a bound on the size of the subset of paths $\Pi' \subseteq \Pi$ whose suffix intersects $\pi_b$.
We have
$$\Omegahat\left(Ld\right) \le \sum \limits_{v \in \pi_b} \texttt{suffdeg}_{\Pi}(v) = \sum \limits_{v \in \pi_b} \texttt{suffdeg}_{\Pi'}(v) < \left|\Pi'\right|\cdot \sigma,$$
and therefore
$$\left|\Pi'\right| \ge \Omegahat\left(\frac{Ld}{\sigma} \right).$$
Now, choose a node $u \in V$ uniformly at random, and let $Q \subseteq \Pi'$ be the subset of these paths that contain $u$ in their prefix.
We have
\begin{align*}
\mathbb{E}[|Q|] &=   \sum \limits_{\pi \in \Pi'} \Pr\left[ u \in \texttt{prefix}(\pi) \right]\\
                &=   \sum \limits_{\pi \in \Pi'} \frac{\lceil |\pi|/4 \rceil}{n}\\
                &\ge \sum \limits_{\pi \in \Pi'} \Omegahat\left(\frac{L}{n}\right) \tag*{$|\pi| \ge \frac{L}{\lambda}$ from Lemma \ref{lem:pathwitness}}\\
                &= \left|\Pi'\right| \cdot \Omegahat\left(\frac{L}{n}\right)\\
                &= \Omegahat\left(\frac{L^2d}{\sigma n}\right).
\end{align*}
Hence there exists a choice of node $u$ that makes $|Q|$ match or exceed this expectation, which satisfies the lemma.
\end{proof}

We now convert this combinatorial structure into our desired lower bound for $\val$.

\begin{proof} [Proof of Lemma \ref{lem:shortcharge}]
Let $u$ be a node and $Q = \{q_1, \dots, q_k\}$ a set of paths as in Lemma \ref{lem:pathstructure}.
Select corresponding nodes $\{v_1, \dots, v_k\}$ such that each $v_i \in \texttt{suffix}(q_i) \cap \pi_b$, and assume without loss of generality that $v_1$ is the first of these nodes along $\pi_b$.
For each node $v_i$, define the \emph{canonical sequence} to be the node sequence of the form:
$$\begin{cases}
(u, v_1, m, v_i) & \text{if } \dist_S(v_1, v_i) = 2, \text{where $m \in M_{v_1}$ is selected such that $(v_1, m), (m, v_i) \in S$,}\\
(u, v_1, v_i) & \text{if } (v_1, v_i) \in S,\\
(u, v_1=v_i) & \text{if } v_1=v_i.
\end{cases}$$

\begin{center}
\begin{tikzpicture}[
  x=1cm,y=1cm,
  dot/.style={circle,fill=black,inner sep=1.6pt},
  >=Latex
]

\node[dot,label=below right:$u$] (u) at (0,0.7) {};

\draw[thick,->] (4,0) -- (4,4.5) node[above] {$\pi_b$};

\draw[thick,->] (-1.0,0.7) to node [pos=0.6, below] {$q_1$} (4,0.7) -- (5.3,0.7);
\node[dot,label=below right:$v_1$] (v1) at (4,0.7) {};

\draw[thick,->] (-0.8,0.34) -- (0,0.7) to node [pos=0.6, below] {$q_2$} (4,1.8) -- (5.3,2.15);
\node[dot,label=below right:$v_2$] at (4,1.8) {};

\draw[thick,->] (-0.6,0.02) -- (0,0.7) to node [pos=0.6, below] {$q_3$} (4,2.9) -- (5.3,3.6);
\node[dot,label=below right:$v_3$] at (4,2.9) {};

\draw[thick,->] (-0.45,-0.18) -- (0,0.7) to node [pos=0.6, below] {$q_4$} (4,4.0) -- (5.2,5.0);
\node[dot,label=below right:$v_4$] (v4) at (4,4.0) {};

\node[dot, red, label=left:$\color{red} m$] (m) at (4,2.4) {};

\draw[red, ->] (u) to [bend right=10] (v1) to [bend right=15] (m) to [bend right=15] (v4);

\node [red, align=center] at (7, 2.7) {canonical sequence for $v_4$\\ is $(u, v_1, m, v_4)$};
\end{tikzpicture}
\end{center}
Let $(s_i, t_i)$ denote the demand pair associated to each path $q_i$, and note that these are pairwise distinct, since all paths in $Q$ contain the node $u$ but paths associated to the same demand pair must be node-disjoint.
For all $i$, since $u \in \texttt{prefix}(q_i)$ and $v_i \in \texttt{suffix}(q_i)$, we have
\begin{align*}
\val_{s_i, t_i}(u, v_i) &\ge \vdist_{s_i, t_i}(s_i, v_i) - \vdist_{s_i, t_i}(s_i, u)\\
&\ge \frac{1}{L} \cdot \left(\frac{3|q_i|}{4} - \frac{|q_i|}{4}\right)\\
&\ge \Omega\left(\frac{1}{\lambda}\right) \tag*{since $|q_i| \ge \frac{L}{\lambda}$ by Lemma \ref{lem:pathwitness}.}
\end{align*}
Applying Lemma \ref{lem:valtrineq} over the canonical sequence for $v_i$, we therefore have
$$\Omega\left(\frac{1}{\lambda}\right) \le \val_{s_i, t_i}(u, v_i) \le \val_{s_i, t_i}(u, v_1) + \val_{s_i, t_i}(v_1, m) + \val_{s_i, t_i}(m, v_i)$$
(this assumes the first case, where the canonical sequence has the form $(u, v_1, m, v_i)$; the other two cases for the canonical sequence are very similar).
Moreover, we have
$$\val_{s_i, t_i}(m, v_i) \le \frac{\sigma}{\lambda^2 L} \le \frac{1}{\lambda^2},$$
since otherwise the charging scheme would remove $v_i$.
Thus we have
$$\max\left\{\val_{s_i, t_i}(u, v_1), \val_{s_i, t_i}(v_1, m)\right\} \ge \Omega\left(\frac{1}{\lambda}\right) - \frac{1}{\lambda^2} \ge \Omegahat(1).$$
This holds for every choice of $i$, and although the midpoint node $m$ may differ between choices of $i$, we always have $m \in M_{v_1}$ with $|M_{v_1}| \le O(\log n)$.
Hence there is a particular edge $e$, either of the form $e=(u, v_1)$ or $e=(v_1, m), m \in M_v$, with the property that $\val_{s_i, t_i}(e) \ge \Omegahat(1)$ for
$$\ge \frac{|Q|}{|M_{v_1}|+1} \ge \Omegahat\left(\frac{L^2 d}{\sigma n}\right)$$
distinct demand pairs $\{(s_i, t_i)\}$.
Summing over these demand pairs, we get
$$\val(e) \ge \Omegahat\left(\frac{L^2 d}{\sigma n}\right),$$
completing the lemma.
\end{proof}

\subsection{Proof of Lemma \ref{lem:pathwitness} \label{sec:pathwitness}}

Here, we construct a path system $R = (V, \Pi)$ satisfying the various technical properties listed in Lemma \ref{lem:pathwitness}.
Although we describe the construction algorithmically, we note that this part of the argument occurs purely in the analysis, so we do not need to consider issues of runtime.

\begin{mdframed}[backgroundcolor=blue!10]
\textbf{Construction of $\Pi_{s, t}$ Sets.} For each demand pair $(s, t) \in P$ we will describe a construction of a corresponding path set $\Pi_{s, t}$; the final path set $\Pi$ is the union of all these sets.
Initially, $\Pi_{s, t} \gets \emptyset$.
We will write $V(\Pi_{s, t})$ to denote the subset of vertices that participate in any path currently in $\Pi_{s, t}$.
Repeat the following process until no more starting paths $\pi$ may be found:
\begin{itemize}
\item Find an $s \leadsto t$ path $\pi$ in $G \setminus X$ that takes $<\frac{1}{4}$ of its total internal weight from vertices already in $\Pi_{s, t}$.
That is,
$$w_{s, t}\bigg(\left(\pi \setminus \{s, t\}\right) \ \cap \ V(\Pi_{s, t})\bigg) < \frac{1}{4}.$$

\item Let $y$ be the last node along $\pi$ with $\vdist_{s, t}(s, y) \le 1$.
Let $\pi'$ be the prefix $\pi[s \leadsto y]$.

\item Let $\pi'' \subseteq \pi'$ be the vertex subsequence obtained by removing all vertices $v \in \pi'$ with $v \in V(\Pi_{s, t})$.

\item Let $\pi''' \subseteq \pi''$ be the vertex subsequence generated by the following process.
Take the first node in $\pi''$ as the first node in $\pi'''$.
Then, repeat while possible: letting $u$ be the most recent node added to $\pi'''$, find the next node $v \in \pi''$ with the property that
$$\vdist_{s, t}(s, u) + \frac{1}{L} \le \vdist_{s, t}(s, v).$$
Then add $v$ as the next node in $\pi'''$.

\item Add $\pi'''$ to $\Pi_{s, t}$ as a new path.
\end{itemize}
\end{mdframed}

Many of the properties claimed in Lemma \ref{lem:pathwitness} are immediate from this construction.
There are two nontrivial claims which are not so immediate, which we prove in the following two lemmas.

\begin{lemma} [Lower Bound on Path Lengths]
For any path $\pi''' = (v_0, \dots, v_k)$ added to a part $\Pi_{s, t}$, we have $k \ge \Omegahat(L)$.
\end{lemma}
\begin{proof}
For brevity of notation, in the following proof we will write $\vdist$ to stand for $\vdist_{s, t}$, and $w$ to stand for $w_{s, t}$, suppressing the subscripts.
Consider a round of the algorithm in which we construct $\pi \supseteq \pi' \supseteq \pi'' \supseteq \pi'''$.
The first step in our proof is to lower bound the total node weight that was deleted when moving from $\pi'$ to $\pi''$.
We consider two places where these deleted nodes can occur:
\begin{itemize}
\item When moving from $\pi'$ to $\pi''$, we delete the entire prefix $\pi'[s \leadsto v_0]$, except for $v_0$.
The total internal weight of this prefix (not counting $s$, and not counting $v_0$) is at least $\vdist(s, v_0)$.

\item For each index $0 \le i \le k-1$ and each node $v_i$, we delete the internal nodes on the subpath $\pi'[v_i \leadsto v_{i+1}]$.
These have weight
\begin{align*}
w\bigg( \pi'[v_i \leadsto v_{i+1}] \setminus \{v_i, v_{i+1}\}\bigg) &\ge \vdist(s, v_{i+1}) - \vdist(s, v_i) - w(v_i)\\
&\ge \vdist(s, v_{i+1}) - \vdist(s, v_i) - \frac{4^{\epoch}}{L} \tag*{since $v_i \notin X$, so $w(v_i) \le \frac{4^{\epoch}}{L}$}\\
&\ge \vdist(s, v_{i+1}) - \vdist(s, v_i) - \Ohat\left(\frac{1}{L}\right) \tag*{since $\epoch \le O\left( \frac{\log n}{\log \log n}\right)$.}
\end{align*}

\item We also delete all internal vertices along the suffix $\pi'[v_k \leadsto y]$.
Arguing identically to the previous case, the total weight of these internal vertices deleted while moving from $\pi'$ to $\pi''$ is at least 
\begin{align*}
w\bigg( \pi'[v_k \leadsto y] \setminus \{v_k, y\}\bigg) &\ge \vdist(s, y) - \vdist(s, v_k) - \Ohat\left(\frac{1}{L}\right).
\end{align*}
Moreover, by definition of $y$ in the construction, we have
\begin{align*}
\vdist(s, y) &\ge 1 - w(y)\\
&\ge 1 - \Ohat\left(\frac{1}{L}\right).
\end{align*}
Substituting this into the previous inequality, we have
\begin{align*}
w\bigg( \pi'[v_k \leadsto y] \setminus \{v_k, y\}\bigg) &\ge 1 - \vdist(s, v_k) - \Ohat\left(\frac{1}{L}\right).
\end{align*}
\end{itemize}

Now, note that the hypotheses on $\pi$ imply that the \emph{total} weight deleted when we move from $\pi'$ to $\pi''$ is less than $1/4$.
So, adding up our lower bounds on the deleted weight from the previous three cases, we have the inequality:
\begin{align*}
\frac{1}{4} &> \vdist(s, v_0) + \left(\sum \limits_{i=0}^{k-1} \vdist(s, v_{i+1}) - \vdist(s, v_{i}) - \Ohat\left(\frac{1}{L}\right)\right) + \left( 1 - \vdist(s, v_k) - \Ohat\left(\frac{1}{L}\right)\right)\\
&= 1 - \Ohat\left(\frac{k}{L}\right).
\end{align*}
Rearranging this inequality gives $k \ge \Omegahat(L)$, as desired.
\end{proof}

\begin{lemma} [System Size]
The final size of the path system satisfies $|\Pi| \ge \Omegahat\left(\sum \limits_{(s, t) \in P} w_{s, t}(V \setminus X)\right)$.
\end{lemma}
\begin{proof}
For a demand pair $(s, t) \in P$, after completing the construction of $\Pi_{s, t}$, consider the weight function
$$w'_{s, t}(v) := \begin{cases}
1 & \text{if } v \in X\\
4w_{s, t}(v) & \text{if } v \in V(\Pi_{s, t}) \setminus X\\
0 & \text{otherwise.}
\end{cases}$$
We make two observations about this weight function.
First, we claim that it is a valid $(s, t)$ fractional cut.
Since every $s \leadsto t$ path $\pi$ in $G$ either contains an internal vertex in $X$ (so $w(\pi) \ge 1$), or by the halting condition on the construction of $\Pi_{s, t}$, it has
$$w\bigg(\pi \cap \left(V(\Pi_{s, t}) \setminus X\right)\bigg) \ge \frac{1}{4},$$
and thus
\begin{align*}
w'(\pi) &= w'\bigg(\pi \cap (V(\Pi_{s, t})\setminus X)\bigg)\\
        &\ge 4w\bigg(\pi \cap (V(\Pi_{s, t})\setminus X)\bigg)\\
        &\ge 1.
\end{align*}
Second, we note that by our construction, every node in $V \setminus X$ with nonzero weight under $w'_{s, t}$ is in exactly one path in $\Pi_{s, t}$.
Thus we have
\begin{align*}
w'_{s, t}(V \setminus X) &= \sum \limits_{\pi \in \Pi_{s, t}} w'(\pi)\\
&\le \sum \limits_{\pi \in \Pi_{s, t}} |\pi| \cdot \Ohat\left(\frac{1}{L}\right)\\
&\le \sum \limits_{\pi \in \Pi_{s, t}} \Ohat(1) \tag*{since $|\pi| \ge \Omegahat(L)$ by previous lemma}\\
&= \Ohat\left(\left|\Pi_{s, t} \right|\right).
\end{align*}
Using this, we can bound
\begin{align*}
|\Pi| &= \sum \limits_{(s, t) \in P} \left| \Pi_{s, t} \right|\\
&\ge \Omegahat\left(\mass\left(\{w'_{s, t}\}\right)\right)\\
&\ge \Omegahat\left(\mass\left(\{w_{s, t}\}\right)\right).
\end{align*}
To explain the last inequality in this chain a bit further, observe that $\{w'_{s, t}\}$ satisfy the conditions of the candidate weights computed by our main algorithm when checking for the end of an epoch; in particular, they form valid fractional cuts in $G \setminus X$, and their maximum weight on nodes in $V \setminus X$ is at most $4$ times that of $\{w_{s, t}\}$ and hence $\le \frac{4^{\epoch}}{L}$.
Thus, if their mass were smaller than that of $\{w_{s, t}\}$ by more than a factor of $\log n$, we would have triggered the epoch end condition at the end of the previous round of the algorithm.
Since we assume that the epoch is still going, the inequality holds.
\end{proof}

\section{Reductions Among Cut Problems \label{sec:reductions}}

Our goal in this section is to show a sequence of reductions among versions of the flow-cut gap problem; in addition to intrinsic interest, these will be useful to simplify the problem before our main proof.
The network of reductions is summarized in Figure \ref{fig:reductions}.
In order to state our reductions, we define the following functions:
\begin{itemize}
\item The \emph{edge flow-cut gap (FCG)}, denoted $\efg(n, W)$, is the least integer such that for every $n$-node edge-weighted and edge-costed graph $G = (V, E, \cost, w)$ with $W = w(E)$, there is an integral cut $X \subseteq E$ for the pairs at edge-weighted distance $\ge 1$ with
$$\cost(X) \le \langle \cost, w \rangle \cdot \efg(n, W).$$
In other words, the edge part of Theorems \ref{thm:main} and \ref{thm:mainweight} state that $\efg(n, W) \le \min\left\{\Ohat\left(W^{1/2}\right), \Ohat\left(n^{1/3}\right)\right\}$. 

\item The \emph{vertex FCG}, denoted $\vfg(n, W)$, is the least integer such that for every $n$-node vertex-weighted and vertex-costed graph $G = (V, E, \cost, w)$ with $W = w(V)$, there is an integral cut $X \subseteq V$ for the pairs at vertex-weighted distance $\ge 1$ with
$$\cost(X) \le \langle \cost, w \rangle \cdot \vfg(n, W).$$
In other words, the vertex part of Theorems \ref{thm:main} and \ref{thm:mainweight} state that $\vfg(n, W) \le \min\left\{\Ohat\left(W^{1/2}\right), \Ohat\left(n^{1/3}\right)\right\}$.

\item The \emph{vertex FCG with unit costs}, denoted $\vufg(n, W)$, is defined as above with the additional constraint that every vertex $v \in V$ has $\cost(v) = 1$.

\item The \emph{vertex FCG with unit costs and uniform weights}, denoted $\vuufg(n, W)$, is defined as above with the additional constraint that all vertices $v \in V$ have the same weight $w(v)=k$ for some value $k$.
\end{itemize}

\begin{figure}[h]
\begin{center}
\begin{tikzpicture}[
  box/.style={
    draw,
    thick,
    rectangle,
    align=center,
    text width=2cm,
    minimum height=1.6cm
  },
  x=1cm, y=1cm
]

\def\xgap{4.3}

\node[box] (vcr) at (0,0) {vertex FCG\\unit costs\\uniform wts\\$\vuufg(n,W)$};

\node[box] (vuf) at (\xgap,0) {vertex FCG\\unit costs\\$\vufg(n,W)$};

\node[box] (vfg) at (2*\xgap,0) {vertex FCG\\$\vfg(n,W)$};

\node[box] (efg) at (3*\xgap,0) {edge FCG\\$\efg(n, W)$};

\draw [{Latex[length=3mm,width=2mm]}-{Latex[length=3mm,width=2mm]}]
(vcr) to node[midway,above,font=\small] {Section \ref{sec:vertreductions}} (vuf);

\draw [{Latex[length=3mm,width=2mm]}-{Latex[length=3mm,width=2mm]}]
(vuf) to node[midway,above,font=\small] {Section \ref{sec:vufgtovfg}} node[midway, below,align=center,font=\small] {$n \leftrightarrow \texttt{poly}(n)$} (vfg);

\draw [{Latex[length=3mm,width=2mm]}-{Latex[length=3mm,width=2mm]}]
(vfg) to node[midway,above,font=\small] {Section \ref{sec:fcgreductions}}
node[midway,below,align=center,font=\small] {$n \leftrightarrow n \log n$}
(efg);

\draw[-{Latex[length=3mm,width=2mm]}]
(vfg) to[out=120,in=60,looseness=3]
node[above,font=\small] {Section \ref{sec:vfgself}, $W$ to $n$ Parametrization}
(vfg);

\end{tikzpicture}
\end{center}
\caption{\label{fig:reductions} The network of reductions proved in this section.  The parameter $W$ represents the total fractional cut weight, $W:=w(V)$ or $W:=w(E)$.}
\end{figure}
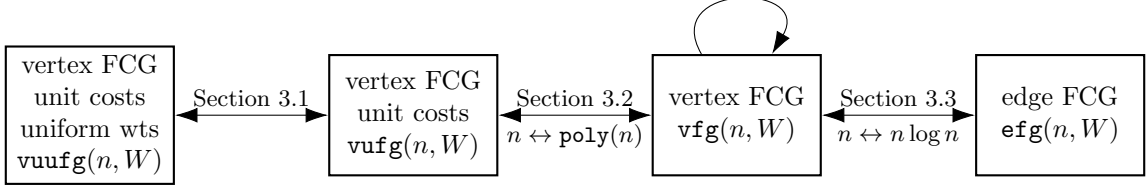

\subsection{\label{sec:vertreductions} Reducing to Uniform Weights}

We first show that, in the setting of unit vertex costs, we can reduce to the setting where we also have uniform vertex weights:

\begin{theorem}
$\vufg(n, W) \le O\left(\vuufg\left(O(n), O(W)\right)\right)$
\end{theorem}
\begin{proof}
Let $G = (V, E, w)$ be an $n$-node input to the vertex unit-cost flow-cut gap problem.
Our goal is to construct a small vertex cut $X \subseteq V$ for the pairs at vertex-weighted distance $\ge 1$.

\paragraph{The Construction.}
Create a new unweighted graph $G' = (V', E')$ by considering each vertex $v \in V$ and replacing it with a path $(v_1, \dots, v_k)$ in $V'$ (with each directed edge $(v_i, v_{i+1}) \in E'$), where
$$k := \left\lceil \frac{w(v)}{w(V)/n} \right\rceil.$$
In other words, if $v$ has the average vertex weight or less then it remains as a single vertex; otherwise we copy it by the number of times it exceeds the average vertex weight (rounded up).
All nodes $v \in V'$ are assigned weight $w'(v) := w(V)/n$.
For each edge $(u, v) \in E$ entering $v$, replace it with $(u, v_1) \in E'$.
For each edge $(v, x) \in E$ leaving $v$, replace it with $(v_k, x) \in E'$.

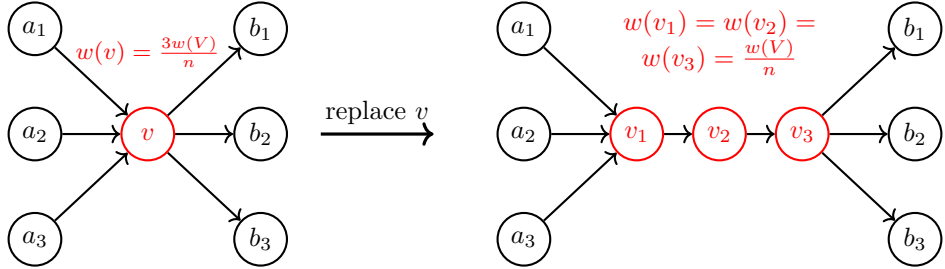
\begin{center}
\begin{tikzpicture}[
  x=1cm,y=1cm,
  node/.style={circle,draw,thick,minimum size=7mm,inner sep=0pt},
  tiny/.style={circle,draw,thick,minimum size=5mm,inner sep=0pt,font=\scriptsize},
  lab/.style={font=\small},
  edge/.style={->,thick},
  fadededge/.style={->,thin,draw=black!30}
]


\def\xA{0.0}
\def\xV{1.5}
\def\xB{3.0}

\node[node] (a1) at (\xA,  1.4) {$a_1$};
\node[node] (a2) at (\xA,  0.0) {$a_2$};
\node[node] (a3) at (\xA, -1.4) {$a_3$};

\node[node,draw=red,text=red] (v) at (\xV,0.0) {$v$};
\node[above=4mm of v, red, align=center] {\small $w(v) = \frac{3w(V)}{n}$};

\node[node] (b1) at (\xB,  1.4) {$b_1$};
\node[node] (b2) at (\xB,  0.0) {$b_2$};
\node[node] (b3) at (\xB, -1.4) {$b_3$};

\draw[edge] (a1) -- (v);
\draw[edge] (a2) -- (v);
\draw[edge] (a3) -- (v);

\draw[edge] (v) -- (b1);
\draw[edge] (v) -- (b2);
\draw[edge] (v) -- (b3);

\draw[->,very thick] (3.8,0) to node[midway,above] {replace $v$} (5.3,0);

\def\xA{6.5}
\def\xVone{8.0}
\def\xVtwo{9.1}
\def\xVthree{10.2}
\def\xB{11.7}

\node[node] (a1r) at (\xA,  1.4) {$a_1$};
\node[node] (a2r) at (\xA,  0.0) {$a_2$};
\node[node] (a3r) at (\xA, -1.4) {$a_3$};

\node[node,draw=red,text=red] (v1r) at (\xVone,0.0) {$v_1$};
\node[node,draw=red,text=red] (v2r) at (\xVtwo,0.0) {$v_2$};
\node[node,draw=red,text=red] (v3r) at (\xVthree,0.0) {$v_3$};
\node[above=3mm of v2r, red, align=center] {$w(v_1)=w(v_2)=$\\$w(v_3)=\frac{w(V)}{n}$};

\node[node] (b1r) at (\xB,  1.4) {$b_1$};
\node[node] (b2r) at (\xB,  0.0) {$b_2$};
\node[node] (b3r) at (\xB, -1.4) {$b_3$};

\draw[edge] (a1r) -- (v1r);
\draw[edge] (a2r) -- (v1r);
\draw[edge] (a3r) -- (v1r);

\draw[edge] (v1r) -- (v2r);
\draw[edge] (v2r) -- (v3r);

\draw[edge] (v3r) -- (b1r);
\draw[edge] (v3r) -- (b2r);
\draw[edge] (v3r) -- (b3r);
\end{tikzpicture}
\end{center}

Letting $n'$ be the number of nodes in $G'$, there is a vertex cut $X'$ for all pairs in $G'$ at weighted distance $\ge 1$, which has size
$$|X'| \le w'(V) \cdot \vuufg\left(n', n' \cdot \frac{w(V)}{n}\right).$$
Now define a vertex cut $X \subseteq V$ by mapping each vertex $x' \in X'$ back to its corresponding vertex in $V$ that created the path containing $x'$.

\paragraph{Correctness.}
We next argue that $X$ is a correct vertex cut.
Consider some pair of nodes $(a, b)$ in $G$ with $\dist_{G \mid w}(a, b) \ge 1$.
Let $(a_k, b_1)$ be the respective last, first nodes of the paths corresponding to $a$ and $b$ in $G'$.
Notice that there is a natural bijection between $a \leadsto b$ paths $\pi$ in $G$ and $a_k \leadsto b_1$ paths $\pi'$ in $G'$, formed by replacing each internal vertex of $\pi$ with its associated path in $G'$.
So, consider one such pair of corresponding paths $\pi, \pi'$.
Let the vertex sequence of $\pi$ be
$$\pi' := (a_k, v_1, \dots, v_j, b_1).$$
Thus we have:
\begin{align*}
w'\left(\pi' \setminus \{a, b\}\right) &\ge \underbrace{\left(\frac{w(V)}{n}\right)}_{\text{weight of each vertex}} \cdot \underbrace{\sum \limits_{z=1}^j \left\lceil \frac{w(v_i) \cdot n}{w(V)} \right\rceil}_{\text{number of internal vertices}}\\
&\ge \sum \limits_{z=1}^j w(v_i)\\
&\ge 1,
\end{align*}
where the last inequality follows since $w$ is a fractional cut for $(a, b)$, and so the sum of vertex weights along $\pi$ is at least $1$.
Thus $X'$ is a cut in $G'$ for $(a_k, b_1)$.
It therefore contains a vertex $v_i \in \pi'$, which is a copy of one of the internal vertices of $\pi$.
Thus $X$ contains a vertex from $\pi$.

\paragraph{Bound on Size.} First, we have
\begin{align*}
n' &= \sum \limits_{v \in V} \left\lceil \frac{w(v) \cdot n}{w(V)} \right\rceil\\
&\le n + \sum \limits_{v \in V} \frac{w(v)\cdot n}{w(V)}\\
&\le 2n.
\end{align*}
Thus, the size of the cut is bounded by
\begin{align*}
|X| &\le |X'|\\
&\le w'(V) \cdot \vuufg\left(n', w'(V)\right)\\
&\le n' \cdot \frac{w(V)}{n} \cdot \vuufg\left(n', n' \cdot \frac{w(V)}{n}\right)\\
&= O\left(w(V)\right) \cdot \vuufg\left(O(n), O(w(V))\right),
\end{align*}
completing the proof.
\end{proof}

\subsection{Reducing to Unit Costs \label{sec:vufgtovfg}}

We next show that we can reduce into the setting of unit vertex costs, at the price of polynomial blowup in the number of vertices:

\begin{theorem}
$\vfg(n, W) \le O(\vufg(O(n^2), O(W)))$
\end{theorem}
\begin{proof}
Let $G = (V, E, \cost, w)$ be an $n$-node instance of the vertex flow-cut gap problem, with $w(V) =: W$.
Our goal is to construct a vertex cut $X \subseteq V$ for the pairs at weighted distance $\ge 1$.

In the following we will say that a \emph{terminal node} is a node in $G$ that is either a source (no incoming edges) or a sink (no outgoing edges).
Cutting a terminal node does not help separate any node pairs, so we may assume without loss of generality that all terminal nodes have weight $0$ in the fractional cut $w$, and also are not involved in any integral cut.
Thus we may change the costs of terminal nodes arbitrarily without affecting the reduction.
So, our focus in the following construction is on modifying the costs of the non-terminal nodes.

\begin{enumerate}
\item The goal of the first step is to ensure that all non-terminal nodes have weights in the range $[\frac{1}{2n}, 1]$.\footnote{This step is necessary only if we want to enforce $\poly(n)$ blowup in the first parameter of the reduction, counting the number of nodes.  In this paper, this is used to ensure that the reduction is algorithmic.  But if one is willing to accept unbounded blowup in the number of nodes, this step can be skipped.}
We achieve this with three steps:
\begin{itemize}
\item While possible, repeat the following step.
Find a non-terminal node $v$ with $w(v) \le \frac{1}{2n}$, and contract it: that is, delete $v$ from the graph, and for each pair of in-neighbors $u$ and out-neighbors $x$ of $v$, we add the edge $(u, x)$.

\item Double the weights of all nodes.
The purpose of this step is to restore the properties of $w$ as a fractional cut, following the previous contraction step.
That is: for any pair of nodes that initially had $\dist_{G \mid w}(u, v) \ge 1$, after the previous contraction step we have $\dist_{G \mid w}(u, v) \ge \frac{1}{2}$, and so after doubling we again have $\dist_{G \mid w}(u, v) \ge 1$.

\item For any nodes with $w(v) > 1$, set $w(v) := 1$.
\end{itemize}

\item By rescaling $\cost$, we may assume without loss of generality that
$$\frac{2\langle \cost, w \rangle}{w(V)} = 1.$$
Then, for all nodes $v \in V$ with $\cost(v) < 1$, increase $\cost(v)$ to $1$.

\item Split each vertex $v \in V$ into $\lceil \cost(v)\rceil$ copies $\{v_i\}$.
For each edge $(u, v)$ (or $(v, u)$), add the corresponding edges $(u, v_i)$ (or $(v_i, u)$) for all these copies.
All copies $\{v_i\}$ are assigned weight $w'(v_i) = w(v)$ as their original node $v$, and cost $1$.
\end{enumerate}

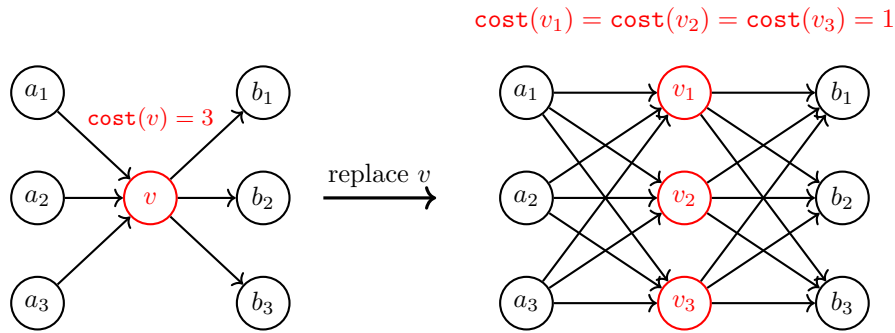
\begin{center}
\begin{tikzpicture}[
  x=1cm,y=1cm,
  node/.style={circle,draw,thick,minimum size=7mm,inner sep=0pt},
  tiny/.style={circle,draw,thick,minimum size=5mm,inner sep=0pt,font=\scriptsize},
  lab/.style={font=\small},
  edge/.style={->,thick},
  fadededge/.style={->,thin,draw=black!30}
]


\def\xA{0.0}
\def\xV{1.5}
\def\xB{3.0}

\node[node] (a1) at (\xA,  1.4) {$a_1$};
\node[node] (a2) at (\xA,  0.0) {$a_2$};
\node[node] (a3) at (\xA, -1.4) {$a_3$};

\node[node,draw=red,text=red] (v) at (\xV,0.0) {$v$};
\node[above=4mm of v, red, align=center] {\small $\cost(v) = 3$};

\node[node] (b1) at (\xB,  1.4) {$b_1$};
\node[node] (b2) at (\xB,  0.0) {$b_2$};
\node[node] (b3) at (\xB, -1.4) {$b_3$};

\draw[edge] (a1) -- (v);
\draw[edge] (a2) -- (v);
\draw[edge] (a3) -- (v);

\draw[edge] (v) -- (b1);
\draw[edge] (v) -- (b2);
\draw[edge] (v) -- (b3);

\draw[->,very thick] (3.8,0) to node[midway,above] {replace $v$} (5.3,0);

\def\xA{6.5}
\def\xV{8.6}
\def\xB{10.7}

\node[node] (a1r) at (\xA,  1.4) {$a_1$};
\node[node] (a2r) at (\xA,  0.0) {$a_2$};
\node[node] (a3r) at (\xA, -1.4) {$a_3$};

\node[node,draw=red,text=red] (v1r) at (\xV,  1.4) {$v_1$};
\node[node,draw=red,text=red] (v2r) at (\xV,  0.0) {$v_2$};
\node[node,draw=red,text=red] (v3r) at (\xV, -1.4) {$v_3$};
\node[above=3mm of v1r, red] {$\cost(v_1)=\cost(v_2)=\cost(v_3)=1$};

\node[node] (b1r) at (\xB,  1.4) {$b_1$};
\node[node] (b2r) at (\xB,  0.0) {$b_2$};
\node[node] (b3r) at (\xB, -1.4) {$b_3$};

\draw[edge] (a1r) -- (v1r);
\draw[edge] (a1r) -- (v2r);
\draw[edge] (a1r) -- (v3r);
\draw[edge] (a2r) -- (v1r);
\draw[edge] (a2r) -- (v2r);
\draw[edge] (a2r) -- (v3r);
\draw[edge] (a3r) -- (v1r);
\draw[edge] (a3r) -- (v2r);
\draw[edge] (a3r) -- (v3r);

\draw[edge] (v1r) -- (b1r);
\draw[edge] (v1r) -- (b2r);
\draw[edge] (v1r) -- (b3r);
\draw[edge] (v2r) -- (b1r);
\draw[edge] (v2r) -- (b2r);
\draw[edge] (v2r) -- (b3r);
\draw[edge] (v3r) -- (b1r);
\draw[edge] (v3r) -- (b2r);
\draw[edge] (v3r) -- (b3r);
\end{tikzpicture}
\end{center}

This completes the construction of $G' = (V', E', w')$.
Let $n'$ be the number of nodes in $G'$.
Since $G'$ has unit costs, it has a cut $X' \subseteq V'$ for all pairs at $w'$-vertex-weighted distance $\ge 1$ of size
$$|X'| \le w'(V') \cdot \vufg(n', w'(V')).$$
Define the cut $X \subseteq V$ by including each vertex $v \in V$ in $X$ iff all of its corresponding copies $\{v_i\}$ are included in $X'$.

\paragraph{Correctness.}

We next argue that $X$ is a correct cut for $G$.
Let $a, b \in V$ with $\dist_{G\mid w}(a, b) \ge 1$.
Let $a_1, b_1 \in V'$ be the first copy of $a, b$, respectively after the node-splitting step of the construction.
For any $a_1 \leadsto b_1$ path $\pi'$ in $G'$, there is a corresponding $a \leadsto b$ path $\pi$ in $G$ found by replacing each internal vertex of $\pi'$ with its corresponding vertex in $\pi$.
Since $\pi, \pi'$ have the same total internal vertex weight, and $\dist_{G \mid w}(a, b) \ge 1$, it follows that $\dist_{G' \mid w'}(a_1, b_1) \ge 1$.
So $X'$ is a cut for $(a_1, b_1)$.

Now, consider any $a \leadsto b$ path $\pi$ in $G$, with vertex sequence
$$\pi = \left(a, v^1, \dots, v^k, b  \right).$$
Each of the nodes $v^j \in \pi$ is expanded to one or more copies $\{v^j_i\}$ in $G'$, and we include all edges between adjacent copies, $\{v^j_i\} \times \{v^{j+1}_i\}$ in $E'$.
Thus, in order to cut $(a_1, b_1)$, it is necessary for $X'$ to include every copy of one of these vertices $v^j$.
Our construction therefore includes $v^j \in X$.

\paragraph{Bound on Size.}
The first step of the construction increases weights by at most a factor of $2$, and so the quantities $w(V)$ and $\langle \cost, w \rangle$ change by at most a constant factor, which may be ignored.

Before the second step of the construction begins, by assumption we have
$$\langle \cost, w \rangle = \frac{w(V)}{2},$$
and so
$$\sum \limits_{v \in V \mid \cost(v) > 1} w(v) \le \frac{w(V)}{2}.$$
We therefore have
$$\sum \limits_{v \in V \mid \cost(v) \le 1} w(v) \ge \frac{w(V)}{2}.$$
When we execute the second step of the construction and round costs up to $1$, we increase the value of $\langle \cost, w \rangle$ by at most
\begin{align*}
&\le \sum \limits_{v \in V} \underbrace{\left(\frac{2\langle \cost, w\rangle}{w(V)}\right)}_{\text{upper bound on increase in $\cost(v)$}} \cdot w(v)\\
&= 2\langle \cost, w \rangle\\
&= \Theta(w(V)),
\end{align*}
so it remains the case that $\langle \cost, w \rangle = \Theta(w(V))$.

In the third step of the construction, each node $v \in V$ contributes $w(v)$ to $w(V)$ before splitting, and it contributes $\lceil \cost(v) \rceil \cdot w(v)$ after splitting.
So the additive increase in $w(V)$ is bounded by
$$\sum \limits_{v \in V} \lceil \cost(v) \rceil \cdot w(v) \le 2\langle \cost, w \rangle \le O(w(V)),$$
which is a constant factor.
Similarly, the value of $\langle \cost, w \rangle$ increases by at most a factor of $2$.
As a final detail, the number of nodes $n'$ after this splitting step is bounded by
\begin{align*}
n' &\le \sum \limits_{v \in V} \lceil \cost(v) \rceil\\
   &\le 4n \cdot \langle \cost, w \rangle \tag*{since weights are $\ge \frac{1}{2n}$}\\
   &= 4n \cdot \Theta(w(V))\\
   &= \Theta(n^2) \tag*{since weights are $\le 1$.}
\end{align*}

We can now put the pieces together.
We have:
\begin{align*}
\cost(X) &\le O(|X'|) \tag*{each vertex $x \in X$ comes from $\lceil \cost(x) \rceil$ vertices in $X'$}\\
&\le O(w'(V)) \cdot \vufg\left(n', w'(V')\right)\\
&= \Theta(\langle \cost, w \rangle) \cdot \vufg\left(\Theta(n^2), \Theta(w(V))\right),
\end{align*}
completing the proof.
\end{proof}

\subsection{\label{sec:fcgreductions} Reducing Between Edge and Vertex Cuts}

\begin{theorem} \label{thm:edgetovert}
$\efg(n, W) \le O\left(\vfg\left(O(n \log n), O(W)\right)\right)$
\end{theorem}
\begin{proof}
Let $G = (V, E, \cost), w$ be an $n$-node instance of the edge flow-cut gap problem, with $w(E) =: W$.
Our goal is to construct an edge cut $X \subseteq E$ for the pairs at weighted distance $\ge 1$.

\paragraph{The Construction.}
We use the following steps to create a new instance $G' = (V', E', \cost')$, where $\cost'$ holds \emph{vertex} costs, as well as a new fractional \emph{vertex} cut $w'$.
We will assume for convenience in the following that $n$ is a power of $2$; otherwise, it may be rounded without issue.

\begin{itemize}
\item For each vertex $v \in V$, replace $v$ with:
\begin{itemize}
\item An $L_v \times R_v$ biclique.
The nodes in both $L_v$ and $R_v$ are labeled from the set
$$\left\{1, \frac{1}{2}, \frac{1}{4}, \dots, \frac{1}{n}, 0\right\}$$
(hence $|L_v| = |R_v| = O(\log n)$).
The label of each vertex will act as its weight in $w'$.

\item Additionally, add a vertex called $v_s$ with an outgoing edge to every node in $R_v$, and add a vertex called $v_e$ with an incoming edge from every node in $L_v$, with $w'(v_s) = w'(v_e) = 0$ (the subscripts $s, e$ stand for \emph{start, end}).
\end{itemize}

\item For each edge $(u, v) \in E$, we map it to an edge $\phi(u, v) := (u'_i, v'_i) \in E'$, where $u'_i \in R_u, v'_i \in L_v$, $i$ denotes the label of these nodes, and:
\begin{itemize}
\item If $w(u, v) \le \frac{1}{2n}$, then $i=0$

\item If $w(u, v) \ge 1$, then $i=1$

\item Otherwise, choose $i$ to be the smallest number that is a power of $\frac{1}{2}$ at least as large as $w(u, v)$.
\end{itemize}
Delete any nodes in a biclique $L_v \cup R_v$ for which we do not map any incident edges.

\item Abusing notation, for $v' \in V'$, we will denote by $\phi^{-1}(v')$ the set of edges $e \in E$ with $\phi(e)$ incident to $v'$.
We assign vertex costs
$\cost'(v') := \cost\left(\phi^{-1}(v')\right).$
\end{itemize}

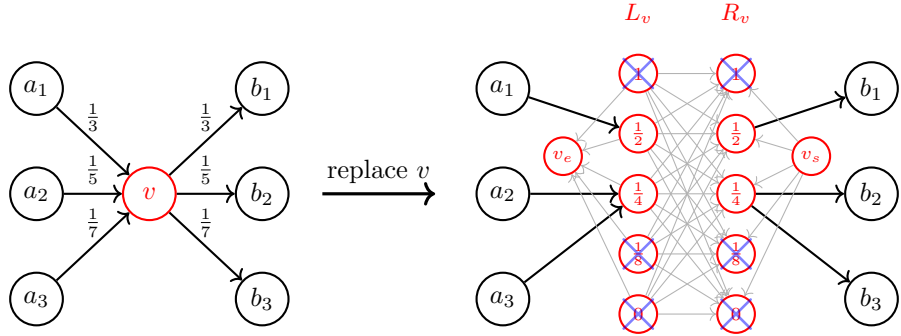
\begin{center}
\begin{tikzpicture}[
  x=1cm,y=1cm,
  node/.style={circle,draw,thick,minimum size=7mm,inner sep=0pt},
  tiny/.style={circle,draw,thick,minimum size=5mm,inner sep=0pt,font=\scriptsize},
  lab/.style={font=\small},
  edge/.style={->,thick},
  fadededge/.style={->,thin,draw=black!30}
]


\def\xA{0.0}
\def\xV{1.5}
\def\xB{3.0}

\node[node] (a1) at (\xA,  1.4) {$a_1$};
\node[node] (a2) at (\xA,  0.0) {$a_2$};
\node[node] (a3) at (\xA, -1.4) {$a_3$};

\node[node,draw=red,text=red] (v) at (\xV,0.0) {$v$};

\node[node] (b1) at (\xB,  1.4) {$b_1$};
\node[node] (b2) at (\xB,  0.0) {$b_2$};
\node[node] (b3) at (\xB, -1.4) {$b_3$};

\draw[edge] (a1) -- node[lab,above] {$\tfrac{1}{3}$} (v);
\draw[edge] (a2) -- node[lab,above] {$\tfrac{1}{5}$} (v);
\draw[edge] (a3) -- node[lab,above] {$\tfrac{1}{7}$} (v);

\draw[edge] (v) -- node[lab,above] {$\tfrac{1}{3}$} (b1);
\draw[edge] (v) -- node[lab,above] {$\tfrac{1}{5}$} (b2);
\draw[edge] (v) -- node[lab,above] {$\tfrac{1}{7}$} (b3);

\draw[->,very thick] (3.8,0) to node[midway,above] {replace $v$} (5.3,0);


\def\xRA{6.2}
\def\xL{8.0}     
\def\xR{9.3}     
\def\xRB{11.1}

\node[node] (A1) at (\xRA,  1.4) {$a_1$};
\node[node] (A2) at (\xRA,  0.0) {$a_2$};
\node[node] (A3) at (\xRA, -1.4) {$a_3$};

\node[node] (B1) at (\xRB,  1.4) {$b_1$};
\node[node] (B2) at (\xRB,  0.0) {$b_2$};
\node[node] (B3) at (\xRB, -1.4) {$b_3$};

\node[tiny,draw=red,text=red] (L1) at (\xL,  1.6) {$1$};
\node[tiny,draw=red,text=red] (L2) at (\xL,  0.8) {$\tfrac12$};
\node[tiny,draw=red,text=red] (L3) at (\xL,  0.0) {$\tfrac14$};
\node[tiny,draw=red,text=red] (L4) at (\xL, -0.8) {$\tfrac18$};
\node[tiny,draw=red,text=red] (L5) at (\xL, -1.6) {$0$};

\node [opacity=0.5] at (L1) {\Huge \color{blue} $\times$};
\node [opacity=0.5] at (L4) {\Huge \color{blue} $\times$};
\node [opacity=0.5] at (L5) {\Huge \color{blue} $\times$};

\node[tiny,draw=red,text=red] (R1) at (\xR,  1.6) {$1$};
\node[tiny,draw=red,text=red] (R2) at (\xR,  0.8) {$\tfrac12$};
\node[tiny,draw=red,text=red] (R3) at (\xR,  0.0) {$\tfrac14$};
\node[tiny,draw=red,text=red] (R4) at (\xR, -0.8) {$\tfrac18$};
\node[tiny,draw=red,text=red] (R5) at (\xR, -1.6) {$0$};

\node [opacity=0.5] at (R1) {\Huge \color{blue} $\times$};
\node [opacity=0.5] at (R4) {\Huge \color{blue} $\times$};
\node [opacity=0.5] at (R5) {\Huge \color{blue} $\times$};

\foreach \i in {1,2,3,4,5}{
  \foreach \j in {1,2,3,4,5}{
    \draw[fadededge] (L\i) -- (R\j);
  }
}

\draw[edge] (A1) -- (L2);
\draw[edge] (A2) -- (L3);
\draw[edge] (A3) -- (L3);

\draw[edge] (R2) -- (B1);
\draw[edge] (R3) -- (B2);
\draw[edge] (R3) -- (B3);

\node [tiny, draw=red, text=red] (VE) at ({\xL-1}, 0.5) {$v_e$};
\node [tiny, draw=red, text=red] (VS) at ({\xR+1}, 0.5) {$v_s$};
\foreach \i in {1,2,3,4,5}{
    \draw[fadededge] (VS) -- (R\i);
    \draw[fadededge] (L\i) -- (VE);
 }

\node[lab,text=red] at (\xL,2.4) {$L_v$};
\node[lab,text=red] at (\xR,2.4) {$R_v$};

\end{tikzpicture}
\end{center}

This new graph $G' = (V', E', \cost')$ has $n' = O(n \log n)$ nodes, since each vertex has been split into $O(\log n)$ vertices.
Additionally, it follows from the construction that $w'(V') \le O(w(E))$.
Hence there is a vertex cut $X' \subseteq V'$, for all pairs of nodes $(u, v)$ with $\dist_{G' \mid w'}(u, v) \ge 1$, and which has
$$\cost'(X') \le \cost'(w') \cdot \vfg\left(O(n \log n), O(W)\right).$$
Then, the final integral edge cut $X \subseteq E$ is defined as
$$X := \phi^{-1}(X') := \bigcup \limits_{x' \in X'} \phi^{-1}(x').$$

\paragraph{Correctness.}

Consider any pair of nodes $(a, b)$ in $G$ with $\dist_{G \mid w}(a, b) \ge 1$.
Let $\pi$ be any $a \leadsto b$ path that is simple (does not repeat nodes).
Let $\pi'$ be the corresponding $a_s \leadsto b_e$ path in $G'$, where each edge $(u, v) \in \pi$ is replaced by its mapped edge $\phi(u, v) \in E'$, and the endpoints of these edges are connected via biclique edges.
Notice that the endpoint nodes of $\phi(u, v)$ each have weight larger than $w(u, v)$, if $w(u, v) > \frac{1}{2n}$.
Thus we have
\begin{align*}
w'(\pi') &\ge 2 \cdot \sum \limits_{\substack{(u, v) \in \pi\\w(u, v) > \frac{1}{2n}}} w(u, v)\\
&= \left(2 \cdot \sum \limits_{(u, v) \in \pi} w(u, v)\right) - \left(2 \cdot \sum \limits_{\substack{(u, v) \in \pi\\w(u, v) \le \frac{1}{2n}}} w(u, v)\right)\\
&= \left(2 \cdot 1\right) - \left(2 \cdot \sum \limits_{\substack{(u, v) \in \pi\\w(u, v) \le \frac{1}{2n}}} w(u, v)\right) \tag*{since $w$ is a frac cut}\\
&\ge 2 \cdot 1 - 2 \cdot \frac{1}{2} \tag*{since $\pi$ is simple, so has at most $n$ edges}\\
&= 1.
\end{align*}
Thus $w'$ is a fractional cut for the pair $(a_s, b_e)$, and so $X'$ contains an internal vertex of $\pi'$.
By construction, this means there is an edge $e \in \pi$ that is mapped incident to $x'$, and so $e \in \phi^{-1}(x')$, and so $e \in X$.
Since we have argued that $X$ contains an edge from an arbitrary $a \leadsto b$ path $\pi$, it follows that $X$ is a valid integral cut.

\paragraph{Bound on $\cost(X)$.}

We have:
\begin{align*}
\cost(X) &= \cost\left( \bigcup \limits_{x' \in X'} \phi^{-1}(x')\right)\\
&\le \sum \limits_{x' \in X'} \cost(\phi^{-1}(x'))\\
&= \sum \limits_{x' \in X'} \cost'(x')\\
&= \cost'(X')\\
&\le \cost'(w') \cdot \vfg\left( O(n \log n), O(W) \right). 
\end{align*}

To continue simplifying, consider any edge $(u, v) \in E$.
This edge contributes $w(u, v) \cdot \cost(u, v)$ to $\langle \cost, w \rangle$.
In $G'$, this edge causes us to add $\cost(u, v)$ to both endpoint nodes of $\phi(u, v)$, which each have weight $O(w(u, v))$.
It follows that $\cost'(w') \le O(\langle \cost, w\rangle)$, and so we have
$$\cost(X) \le \langle \cost, w \rangle \cdot O\left(\vfg\left( O(n \log n), O(W) \right)\right),$$
completing the proof.
\end{proof}

\begin{theorem} \label{thm:vfgleefg}
$\vfg(n, W) \le \efg(O(n), W)$
\end{theorem}
\begin{proof}
The following reduction is standard in the literature, and described e.g.\ in \cite{Gupta03, CK09}, etc., so we will describe it only briefly.
Let $G = (V, E, \cost), w$ be an $n$-node input to the vertex flow-cut gap problem, with $w(V) =: W$.
Our goal is to construct a vertex cut $X \subseteq V$ for the pairs at vertex-weighted distance $\ge 1$.

Construct a new graph $G' = (V', E', \cost')$ by considering each vertex $v \in V$ in arbitrary order, and performing the following steps:
\begin{itemize}
\item Replace $v$ with two new vertices, $v_{in}$ and $v_{out}$.

\item Add the directed edge $(v_{in}, v_{out})$.
Assign its weight in $w'$ to be $w(v)$, and its cost in $\cost'$ to be $\cost(v)$.

\item For all edges $(u, v)$ entering $v$, replace them with $(u, v_{in})$, and for all edges $(v, u)$ leaving $v$, replace them with $(v_{out}, u)$.
Assign both types of edges weight $0$ under $w'$, and some sufficiently high cost under $\cost'$.
\end{itemize}
This graph has $n' = 2n$ nodes, and $w'(E)=W$.
So, there is an integral edge cut $X' \subseteq E'$ for the pairs $(u, v)$ with $\dist_{G' \mid w'}(u, v) \ge 1$, with
$$\cost'(X') \le \langle \cost, w' \rangle \cdot \efg(2n, W).$$
Moreover, by choice of high enough costs, all edges in $X'$ must be of the form $(v_{in}, v_{out})$ for some $v \in V$.
Define the vertex cut $X$ to contain all such nodes $v$.
Now, a straightforward analysis (that we omit) shows that the vertex cut $X$ satisfies the theorem.
\end{proof}

\subsection{A Self-Reduction for $\vfg$ \label{sec:vfgself}}

The following self-reduction is implied by a simple node-cutting technique, which is standard in the literature and has been used in most prior work \cite{AAC07, CK09, Gupta03}.

\begin{theorem} [Folklore] \label{thm:vfgself}
If $\vfg(n, W) \le W^{c+o(1)}$ for some constant $c$, then $\vfg(n, W) \le n^{\frac{c}{1+c} + o(1)}$.
\end{theorem}
\begin{proof}
Let $G = (V, E, \cost, w)$ be an $n$-node input to the vertex flow-cut gap problem.

\paragraph{The Construction.}
We construct an integral cut $X \subseteq V$ for the pairs at vertex-weighted distance $\ge 1$ using the following steps.

\begin{itemize}

\item Let $X_1$ be the set of all nodes of fractional weight $\ge n^{-\frac{c}{1+c}}/4$.

\item Let $2w$ denote the vertex weight function that assigns every vertex twice its weight in $w$.
Let $X_2$ be an integral cut in the graph $G \setminus X_1$, under weight function $2w$, for all pairs $(a, b)$ at distance $\dist_{G \setminus X_1 \mid 2w}(a, b) \ge 1$.

\item The final cut is $X := X_1 \cup X_2$.
\end{itemize}

\paragraph{Correctness.}
Consider any node pair $(a, b)$ with $\dist_{G \mid w}(a, b) \ge 1$, and let
$\pi = (a, v_1, \dots, v_k, b)$
be any $a \leadsto b$ path.
If any of the internal vertices $v_1, \dots, v_k$ are added to $X_1$ in the first step of the construction, then we are done, so let us assume otherwise.
Our plan is to then argue that one of the vertices $v_2, \dots, v_{k-1}$ are added to $X_2$.

Since we assume $v_1, v_k \notin X_1$, we must have
$$w(v_1), w(v_k) \le \frac{n^{-\frac{c}{1+c}}}{4}.$$
By the triangle inequality, we therefore have
\begin{align*}
\dist_{G \mid w}(v_1, v_k) &\ge \dist(a, b) - w(v_1) - w(v_k)\\
&\ge 1 -  \frac{n^{-\frac{c}{1+c}}}{4} -  \frac{n^{-\frac{c}{1+c}}}{4}\\
&\ge \frac{1}{2},
\end{align*}
and so
$$\dist_{G \setminus X_1 \mid 2w}(v_1, v_k) \ge 2 \cdot \dist_{G \mid w}(v_1, v_k) \ge 1.$$
So $X_2$ is an integral cut for the pair $(v_1, v_k)$, and so it includes at least one node from $v_2, \dots, v_{k-1}$.
Thus $X$ includes an internal vertex of $\pi$.

\paragraph{Bound on $\cost(X)$.}
First, we have:
\begin{align*}
\cost(X_1) &= \sum \limits_{v \in V \mid w(v) \ge n^{-\frac{c}{c+1}}/4} \cost(v)\\
&\le \sum \limits_{v \in V} \cost(v) \cdot w(v) \cdot n^{\frac{c}{c+1}}/4\\
&= \langle \cost, w \rangle \cdot O\left(n^{\frac{c}{c+1}}\right).
\end{align*}
Then, we have:
\begin{align*}
\cost(X_2) &\le \langle \cost, 2w\rangle \cdot \left( 2w(V \setminus X_1)\right)^{c+o(1)}\\
&\le \langle \cost, w\rangle \cdot \left( n \cdot n^{-\frac{c}{c+1}}\right)^{c+o(1)} \tag*{max weight in $V \setminus X_1$ is $\le n^{-\frac{c}{c+1}}$}\\
&= \langle \cost, w \rangle \cdot \left(n^{\frac{1}{c+1}}\right)^{c+o(1)}\\
&= \langle \cost, w \rangle \cdot n^{\frac{c}{c+1} + o(1)}.
\end{align*}
Combining these two bounds completes the proof.
\end{proof}

\subsection{Reducing Between Flow-Cut Gap and Weak Low-Diameter Decompositions \label{sec:wlddred}}

Here we provide details on the proof of Corollary \ref{cor:wldd}.
In fact, we will phrase it more generally as a bound for weak low-diameter decompositions (LDDs).
In the following theorem statement and proof, we use $\Oish(\cdot)$ notation to hide factors of $\log n$ \emph{or} $\log W$.

\begin{theorem}
There is a weak LDD for vertex cuts with loss parameter $\Oish(\vfg(n, W))$, and a weak LDD for edge cuts with loss parameter $\Oish(\efg(n, W))$.
\end{theorem}
\begin{proof}
We will show the proof here for edge cuts; the proof for vertex cuts is essentially identical.
The input is an $n$-node directed edge-weighted graph $G = (V, E, w)$.
We will assume without loss of generality that the minimum nonzero weight is at least $\frac{1}{2n}$.
This may be achieved, e.g., by rounding all lower weights down to $0$, while doubling all remaining weights.
Note that any pair $(s, t)$ with $\dist(s, t) \ge 1$ will still have this property following this weight modification.

Our goal will be to construct\footnote{We note that our construction of $\mathcal{X}$ runs in randomized polynomial time, so long as the subroutine computing a cut satisfying the $\efg(n, W)$ bound is as well.} a multiset of cuts $\mathcal{X}$ of size $|\mathcal{X}| \le \Oish(n)$, where:
\begin{itemize}
\item each cut $X \in \mathcal{X}$ deterministically hits all $s \leadsto t$ paths for any pair $(s, t)$ with $\dist_G(s, t) \ge 1$, and
\item each edge $e \in E$ is in at most $w(e) \cdot \Oish(\efg(n, W)) \cdot |\mathcal{X}|$ many cuts in $\mathcal{X}$.
\end{itemize}
Thus, uniformly sampling a cut from $\mathcal{X}$ suffices for a weak low-diameter decomposition.
To construct $\mathcal{X}$, we repeat the following procedure:

\begin{itemize}
\item initially, assign the \emph{cost} of each edge with nonzero weight to be $1$, and the cost of each edge with zero weight to be a very large value (such that these edges will never be cut).
\item in each round $i$, we may compute an edge cut $X_i$ that satisfies
$\cost(X_i) \le \efg(n, W) \cdot \langle \cost, w \rangle,$
with respect to the current edge costs.
Then we add $X_i$ to $\mathcal{X}$.
\item after computing $X_i$, for each edge $e \in X_i$, update
$$ \cost(e) \gets \cost(e) \cdot \left( 1 + \frac{1}{w(e) \cdot \efg(n, W)}\right),$$
where $W := w(E)$.
The costs of edges that are not in $X_i$ remain unchanged.
\end{itemize}

Our goal is now to prove that each edge $e \in E$ is in at most $w(e) \cdot \Oish(\efg(n, W))$ many cuts in $\mathcal{X}$.
We first observe that in each round $i$, when we update $\cost$, the quantity $\langle \cost, w \rangle$ increases to at most
\begin{align*}
& \langle \cost, w \rangle + \sum_{e \in X_i} \cost(e) \cdot w(e) \cdot \frac{1}{w(e) \cdot \efg(n, W)}\\
= \ & \langle \cost, w \rangle  + \frac{\cost(X_i)}{\efg(n, W)}\\
\le& \ 2\cdot \langle \cost, w\rangle.
\end{align*}

Since initially $\langle \cost, w \rangle = W$, all cuts $X_i \in \mathcal{X}$ must satisfy
$$\cost(X_i) \le 2^i \cdot W \cdot \efg(n, W).$$
Meanwhile, suppose $e$ is in at least $k_e+1$ cuts.
Then its cost in the last round $i$ in which it is cut must satisfy
\begin{align*}
\left( 1 + \frac{1}{w(e) \cdot \efg(n, W)}\right)^{k_e} = \cost(e) \le \cost(X_i) \le 2^i \cdot W \cdot \efg(n, W).
\end{align*}
We then compare the first and last terms and rearrange, as follows:
\begin{align*}
\left(1+\frac{1}{w(e)\cdot \efg(n,W)}\right)^{k_e}
&\le 2^i \cdot W \cdot \efg(n, W) \\
k_e \log\left(1+\frac{1}{w(e)\cdot \efg(n,W)}\right)
&\le i\log 2 + \log W + \log \efg(n, W) \\
\frac{k_e}{i}
&\le
\frac{
\Oish(1)
}{
\log\left(1+\frac{1}{w(e)\cdot \efg(n,W)}\right)
}.
\end{align*}
At this point, we note that we may assume that $w(e) \cdot \efg(n, W) \le 1$, since otherwise the claim is satisfied trivially.
So we may apply the estimate $\log(1+\frac{1}{x}) = \Theta(1/x)$, and continue simplifying:
\begin{align*}
\frac{k_e}{i}
&\le
\frac{
\Oish(1)
}{
\Theta\left(\frac{1}{w(e)\cdot \efg(n,W)}\right)} \\
\frac{k_e}{i}
&\le
\Oish\left(
w(e)\cdot \efg(n,W)
\right)\\
\frac{k_e + 1}{i}
&\le
\Oish\left(
w(e)\cdot \efg(n,W)
\right) + \frac{1}{i}.
\end{align*}
Thus the desired bound holds so long as $\frac{1}{i} \le \Oish\left(
w(e)\cdot \efg(n,w(E))
\right)$.
Since $w(e) \ge \frac{1}{2n}$, this holds when $i \ge \Omega(n)$.
So, repeating for that many rounds provides our desired guarantee.
\end{proof}

\section*{Acknowledgments}

We are grateful to Gary Hoppenworth for useful technical discussions, and for some helpful references to prior work.

\bibliographystyle{plain}
\bibliography{refs}

\appendix

\end{document}